\documentclass[11pt]{article}
\usepackage{amssymb, amsmath, amsthm, url, times}
\usepackage[noadjust]{cite}

\usepackage[margin=1in]{geometry}
\usepackage{amsfonts,amsmath,color,float,graphicx,verbatim}

\usepackage{enumerate}
    
\title{Hardness Amplification and the Approximate Degree 
of Constant-Depth Circuits}
\date{}

\author{Mark Bun \thanks{Harvard University, School of Engineering and Applied Sciences. Supported by an NDSEG Fellowship and NSF grant CNS-1237235.} \\ \texttt{mbun@seas.harvard.edu}
   \and 
Justin Thaler \thanks{Simons Institute for the Theory of Computing at UC Berkeley. Parts of this work were performed while the author was a graduate student at Harvard University, School of Engineering and Applied Sciences. This work was supported by an NSF Graduate Research Fellowship, NSF grants CNS-1011840 and CCF-0915922, and a Research Fellowship from the Simons Institute for the Theory of Computing.}   \\ \texttt{jthaler@fas.harvard.edu}}

\usepackage{amsfonts,amsmath,color,float,graphicx,verbatim}

\usepackage{enumerate}

\usepackage{thmtools}
\usepackage{thm-restate}

\newcommand{\R}{\mathbb{R}}

\newcommand{\E}{\mathbb{E}}

\newcommand{\sgn}{\widetilde{\mathrm{sgn}}}
\newcommand{\PP}{\operatorname{PP}}

\newcommand{\poly}{\mathrm{poly}}

\newcommand{\bits}{\{-1,1\}}

\newcommand{\eat}[1]{}

\newcommand{\omb}{{\sc ODD-MAX-BIT}}
\newcommand{\ombmath}{${\sc ODD-MAX-BIT}$}
\newcommand{\ED}{{\sc Element Distinctness}}
\newcommand{\tto}{{\sc 2-to-1}}
\newcommand{\ttov}{{\sc 2-to-1-VS-Almost-2-to-1}}
\newcommand{\tdeg}{\deg_{\pm}}
\newcommand{\AND}{\operatorname{AND}}
\newcommand{\OR}{\operatorname{OR}}
\newcommand{\ANDOR}{\operatorname{AND-OR}}
\newcommand{\ORAND}{\operatorname{OR-AND}}

\newcommand{\etal}{{et al.}}

\usepackage{color}

\newcommand{\ignore}[1]{}
\newcommand{\provisionallyremove}[1]{}

\newcommand{\eps}{\varepsilon}

\newcommand{\disc}{\operatorname{disc}}
\newcommand{\acz}{AC$^{\mbox{0}}$}

\newcommand{\secref}[1]{Section \ref{#1}}

\renewcommand{\eqref}[1]{Eq.~(\ref{#1})}
\newcommand{\lemref}[1]{Lemma~\ref{#1}}
\newcommand{\corref}[1]{Corollary~\ref{#1}}
\newcommand{\thmref}[1]{Theorem \ref{#1}}

\newcommand{\appref}[1]{Appendix~\ref{#1}}

\newcommand{\odeg}{\widetilde{\operatorname{odeg}}}
\newcommand{\adeg}{\widetilde{\operatorname{deg}}}
\newcommand{\sel}{\operatorname{Sel}}

\newtheorem{theorem}{Theorem}
\newtheorem{lemma}[theorem]{Lemma}
\newtheorem{proposition}[theorem]{Proposition}
\newtheorem{corollary}[theorem]{Corollary}

\newtheorem{fact}[theorem]{Fact}

\newtheorem{claim}[theorem]{Claim}

\newtheorem{innercustomthm}{Theorem}
\newenvironment{customthm}[1]
  {\renewcommand\theinnercustomthm{#1}\innercustomthm}
  {\endinnercustomthm}

\begin{document}

\maketitle

\vspace{-6mm}
\begin{abstract}
We establish a generic form of hardness amplification for the approximability of constant-depth Boolean circuits by polynomials. Specifically, we show that if a Boolean circuit cannot be 
pointwise approximated by low-degree polynomials to within constant error 
in a certain one-sided sense, then an OR of disjoint copies of that circuit cannot be pointwise approximated even with very high error.
As our main application, we show that for every sequence of degrees $d(n)$, there is an explicit depth-three circuit $F: \bits^n \to \bits$ of polynomial-size such that any degree-$d$ polynomial cannot pointwise approximate $F$ to error better than $1-\exp\left(-\tilde{\Omega}(nd^{-3/2})\right)$. 
As a consequence of our main result, 
we obtain an $\exp\left(-\tilde{\Omega}(n^{2/5})\right)$ upper bound on the the discrepancy of a function in \acz, and an $\exp\left(\tilde{\Omega}(n^{2/5})\right)$ lower bound on the threshold weight of \acz, 
improving over the previous best results of $\exp\left(-\Omega(n^{1/3})\right)$ and
$\exp\left(\Omega(n^{1/3})\right)$ respectively. 

Our
techniques also yield a new lower bound of $\Omega\left(n^{1/2}/\log^{(d-2)/2}(n)\right)$ on the approximate
degree of the AND-OR tree of depth $d$, which is tight up to polylogarithmic
factors for any constant $d$,
as well as new bounds for read-once DNF formulas.
In turn, these results imply new lower bounds on the communication and circuit complexity of these classes, 
and demonstrate strong limitations on existing PAC learning algorithms.

\end{abstract}

\section{Introduction}
The $\eps$-approximate degree of a Boolean function $f: \{-1,1\}^n \rightarrow \{-1, 1\}$,
denoted $\adeg_\eps(f)$,
is the minimum degree of
a real polynomial that approximates $f$ to error $\eps$ in the 
$\ell_{\infty}$ norm. Approximate degree has pervasive applications
in theoretical computer science. For example, lower bounds on
approximate degree underly 
many tight lower bounds on quantum query complexity (e.g., \cite{aaronsonshi, beals, ambainis, klauck, sherstov-direct}), and
have been used to resolve several long-standing open questions in communication complexity \cite{patternmatrix, patternmatrixfollowup, ada, spalek, bvdw, shizhu, sherstov-disjoint, sherstov-disjoint2} (see the survey paper by Sherstov \cite{sherstovsurvey}). Meanwhile,
upper bounds on approximate degree underly 
many of the fastest known learning algorithms,
including PAC learning DNF and read-once formulas \cite{ks, spalekandor}, agnostically learning disjunctions \cite{kalaiagnostic},
and PAC learning in the presence of irrelevant information \cite{ksomb, thalercolt}.  

Despite the range and importance of these applications, large gaps remain in our understanding of approximate degree. The approximate degree of any \emph{symmetric} Boolean function has been understood since Paturi's 1992 paper \cite{paturi}, but once we move beyond symmetric functions, few general results are known.

In this paper, we perform a careful study of the approximate 
degree of constant-depth Boolean circuits. In particular,
we establish a generic form of hardness amplification for the pointwise approximation of
small depth circuits by low-degree polynomials: we show that if a Boolean circuit $f$ cannot be 
pointwise approximated to within constant error 
in a certain one-sided sense by polynomials of a given degree, then the circuit $F$ obtained by taking an OR of disjoint copies of $f$ cannot be  approximated even with error exponentially close to 1. Notice that if $f$ 
is computed by a circuit of polynomial size and constant depth, then so is $F$.

Our proof extends
a recent line of work \cite{sherstovFOCS, lee, bunthaler, sherstovnew}
that seeks to prove approximate degree lower bounds by constructing explicit 
\emph{dual polynomials}, which are dual solutions
to a linear program that captures the approximate degree of any function. Specifically, we show that given a dual polynomial 
demonstrating that $f$ cannot be approximated to within constant
error, we can construct a dual polynomial demonstrating
that $F$ cannot be approximated even with error exponentially close
to 1.

As the main application of our
hardness amplification technique, for any $d>0$
we exhibit an explicit function $F: \bits^n \to \bits$ computed by a polynomial size circuit of depth three for which any degree-$d$ polynomial cannot pointwise approximate $F$ to error $1-\exp\left(-\tilde{\Omega}(nd^{-3/2})\right)$. We then use this result
to obtain new bounds on two quantities
that play central roles in learning
theory, communication complexity, and circuit complexity:
\emph{discrepancy} and \emph{threshold weight}.
Specifically, we prove a new upper bound of $\exp\left(-\tilde{\Omega}(n^{2/5})\right)$ for the discrepancy
of a function in \acz, and a new lower bound of $\exp\left(\tilde{\Omega}(n^{2/5})\right)$ for the threshold weight of \acz. As a second application, our hardness amplification result allows us to resolve, up to polylogarithmic factors, the approximate degree of AND-OR trees of arbitrary constant depth. Finally, our techniques also yield new lower bounds for read-once DNF formulas.

\section{Overview of Results and Techniques}
\label{sec:overview}
This section provides an overview of our results and the techniques we use to establish them.
We defer detailed proofs to later sections.

\subsection{Hardness Amplification}

Recall that the $\eps$-approximate degree of a Boolean function $f$
is the minimum degree of a real polynomial that pointwise approximates $f$ to error $\eps$. Another fundamental
measure of the complexity of $f$ is its \emph{threshold degree},
denoted $\deg_\pm(f)$. The threshold degree of $f$ is
the least degree of a real polynomial that agrees in sign with $f$ at all Boolean inputs.

Central to our results is a measure of the complexity of a Boolean function that we call \emph{one-sided approximate degree}. This quantity, which we denote by $\odeg_{\eps}(f)$, is an intermediate complexity measure that lies between $\eps$-approximate degree and threshold degree. Unlike approximate degree and threshold degree, one-sided approximate degree
treats inputs in $f^{-1}(1)$ and inputs in $f^{-1}(-1)$ asymmetrically. 

More specifically, $\odeg_\eps(f)$ captures the least degree of a \emph{one-sided approximation} for $f$. Here,
a one-sided approximation $p$ for $f$ is a polynomial that approximates $f$ to error at most $\eps$
at all points $x \in f^{-1}(1)$,  and satisfies the threshold condition $p(x) \leq -1+\eps$ at all points $x \in f^{-1}(-1)$.
Notice that $\odeg_{\eps}(f)$ is always \emph{at most}  $\adeg_{\eps}(f)$, but can be much smaller. Similarly, $\odeg_{\eps}(f)$ is always \emph{at least} $\deg_\pm(f)$, but can be much larger.


One-sided approximate degree is the
complexity measure that we amplify for constant-depth circuits:
given a depth $k$ circuit $f$ on $m$ variables that has one-sided approximate degree greater than $d$, we show how to generically transform $f$ 
into a depth $k+1$ circuit $F$ on $t \cdot m$ variables such that $F$ cannot be pointwise approximated by degree $d$ polynomials
even to error $1-2^{-t}$.\footnote{Follow-up work by Sherstov \cite{sherstovmp} has established a lower bound on the \emph{threshold degree} of $F$. Specifically, he has shown that there is some constant $c$ such that $\deg_\pm(F) > \min\{ct, d\}$. See Section \ref{sec:discussion} for further discussion of this result.}

\begin{theorem}\label{thm:big}
Suppose $f: \bits^m \to \bits$ has one-sided approximate degree $\odeg_{1/2}(f) > d$. Denote by $F: \bits^{m \cdot t} \to \bits$ the block-wise composition $\OR_t(f, \dots, f)$, where $\OR_t$ denotes the OR function on $t$ variables. Then $F$ cannot be pointwise approximated by degree-$d$ polynomials even to within error $1-2^{-t}$ by degree-$d$ polynomials. That is, the $(1-2^{-t})$-approximate degree of $F$ is greater than $d$.
\end{theorem}

\noindent \textbf{Remark:} Theorem \ref{thm:big} demonstrates that $\odeg(f)$ admits a form of hardness amplification within \acz, which does not generally hold for the ordinary approximate degree.
Indeed, Theorem \ref{thm:big} fails badly if the condition $\odeg_{1/2}(f)> d$ is replaced
with the weaker condition $\adeg_{1/2}(f) > d$ (in fact, $f=\OR_m$ is a counter-example; see the discussion in Section \ref{sec:highlow} for details). 
\medskip

A \emph{dual formulation} of one-sided approximate degree
was previously exploited by Gavinsky and Sherstov
to separate the multi-party communication
versions of NP and co-NP \cite{gavinskysherstov}, 
as well as
by the current authors \cite{bunthaler} and independently by Sherstov \cite{sherstovnew} to resolve the approximate degree of the two-level AND-OR tree.
In this paper, we introduce the primal formulation of
one-sided approximate
degree, 
which allows us to express Theorem \ref{thm:big} as a hardness amplification result. 
We also argue for the importance of one-sided approximate degree as a Boolean function complexity measure in its own right.

\eat{
Indeed, subsequent work by Sherstov (discussed further in Section \ref{sec:sherstov}) has strengthened Theorem \ref{thm:big} to show that one-sided approximate degree is the natural complexity to amplify into \emph{threshold degree} lower bounds. Moreover, Kanade and Thaler \cite{kanadethaler} have already shown algorithmic uses of one-sided approximate degree in the framework of \emph{reliable agnostic learning}. In particular, they show that upper bounds on one-sided approximate degree  imply fast algorithms for the model of negative reliable learning, which captures learning in settings where false negatives are costlier than false positives. In light of these recent developments, we believe that the notion of one-sided approximate degree will continue to enable progress on the analysis of Boolean functions.
}




\medskip
\noindent \textbf{Prior Work on Hardness Amplification for Approximate Degree.} 
For the purposes of this discussion, we informally consider
a hardness amplification result
for approximate degree to be any
statement of the following form:
Fix two functions $f: \{-1, 1\}^{m} \rightarrow \{-1, 1\}$ and $g:\{-1, 1\}^t \rightarrow \{-1, 1\}$.
Then the composed function $g(f, \dots, f):\{-1, 1\}^{m\cdot t} \rightarrow \{-1, 1\}$ is strictly harder to approximate
in the $\ell_{\infty}$ norm by low-degree polynomials than is the original function $f$.

We think of such a result as establishing
that application of the outer function $g$ to $t$ disjoint copies of $f$ 
amplifies the hardness of $f$. Here we consider polynomial degree to be a resource, and 
``harder to approximate'' can refer either to
the amount of resources required for the approximation, to the error of the approximation,
or to a combination of the two. 

Two particular kinds 
of hardness amplification results for approximate degree have received
particular attention.
\emph{Direct-sum} theorems focus on amplifying the degree required
to obtain an approximation, but do not focus on amplifying the error. For example,
a typical direct-sum theorem  identifies
conditions on $f$ and $g$ that guarantee 
that $\adeg_{\eps}(g(f, \dots, f)) \geq \adeg_{\eps}(g) \cdot \adeg_{\eps}(f)$. 
In contrast, a \emph{direct-product} theorem focuses on amplifying both the
error and the minimum degree required to achieve this error. An \emph{XOR lemma} is a special case of either type of theorem where the combining function $g$ 
is the XOR function. Ideally, an XOR lemma of the direct-product form
establishes that 
there exists a sufficiently small
constant $\delta > 0$ such that $\adeg_{1-2^{-\delta t}}(\text{XOR}_t(f, \dots, f)) \geq t \cdot \adeg_{1/3}(f)$.
That is, an XOR lemma establishes that approximating the XOR of $t$ disjoint copies of $f$ requires a $t$-fold blowup in degree
relative to $f$, even if one allows error exponentially close to 1.

O'Donnell and Servedio \cite{os} proved an XOR lemma for \emph{threshold degree}, establishing
that $\text{XOR}_t(f, \dots, f)$ has threshold degree $t$ times the threshold degree of $f$.
In later work, Sherstov \cite{sherstovFOCS} proved a direct sum result for approximate degree
that holds
whenever the combining function $g$ has low block-sensitivity. His techniques also capture
O'Donnell and Servedio's XOR lemma for threshold degree as a special case. 
In \cite{sherstov-direct}, Sherstov 
 proved a number of hardness amplification results for approximate degree. 
 Most notably, he proved an optimal XOR lemma, as well
 as  a direct-sum theorem
 that holds whenever the combining function has close to maximal approximate degree (i.e., approximate degree $\Omega(t)$). 
 Sherstov used his XOR lemma to prove direct product theorems for quantum query complexity, and 
in subsequent work \cite{sherstov-disjoint}, to show direct product theorems for the multiparty communication of set disjointness. 

\paragraph{Comparison to Prior Work.} 
In this paper, we are interested in establishing approximate degree lower bounds
for constant-depth circuits over the basis $\{\AND, \OR, \text{NOT}\}$. For this purpose,
it is essential to consider combining functions (such as OR, see Theorem \ref{thm:big}) 
that are themselves in \acz, ruling out
the use of XOR as a combining function. 
Our hardness amplification result (Theorem \ref{thm:big}) 
is orthogonal to direct-sum theorems: direct-sum theorems focus on amplifying
degree but not error, while Theorem \ref{thm:big} focuses on amplifying
error but not degree. Curiously, Theorem \ref{thm:big} is nonetheless 
a critical ingredient in our proof of
a direct-sum type theorem for AND-OR trees of constant depth (Theorem \ref{thm:andor}).

\paragraph{Proof Idea.}

As discussed in the introduction, our proof of Theorem \ref{thm:big} relies on a dual characterization of one-sided approximate degree (cf. Theorem \ref{thm:oprelim}). 
Specifically, for any $m$-variate Boolean function $f$ satisfying $\odeg_{1/2}(f) > d$, 
there exists a dual object $\psi: \{-1, 1\}^m \rightarrow \mathbb{R}$ that witnesses this fact --- we refer to $\psi$ as a ``dual polynomial'' for $f$. As
we show in Theorem \ref{thm:oprelim}, $\psi$ satisfies three important properties: (1) $\psi$ has high correlation with $f$, (2)
$\psi$ has zero correlation with all polynomials of degree at most $d$, and (3) $\psi(x)$ agrees in sign with $f(x)$ for all $x \in f^{-1}(-1)$. 
We refer to the second property by saying that $\psi$ has \emph{pure high degree} $d$, and we refer to the third
property by saying that  $\psi$ has \emph{one-sided error}.

Our proof proceeds by taking a dual witness $\psi$
to the high one-sided approximate degree of $f$,
and a certain dual witness
$\Psi$ for the function $\OR_t$,
and combining them 
to obtain a dual witness $\zeta$ for the fact that 
$\adeg_{1-2^{-t}}(\OR_{t}(f, \dots, f)) > d$.
Our analysis of the combined dual witness crucially
exploits two properties: 
first, that $\psi$ has one-sided error and second, that
the vector whose entries are all equal to $-1$ has very large (in fact, maximal) Hamming distance
from the unique input in $\OR_t^{-1}(1)$.

Our method of combining the two dual witnesses
was first introduced by Sherstov \cite[Theorem 3.3]{sherstovFOCS} and independently by Lee \cite{lee}. This method has also been exploited
by the present authors in \cite{bunthaler} to resolve the approximate degree
of the two-level AND-OR tree, and by Sherstov \cite{sherstov-direct}
to prove direct sum and direct product theorems for polynomial
approximation. However, as discussed above, prior work used this method of combining dual witnesses 
exclusively to
amplify the \emph{degree} in the resulting lower bound; 
in contrast, we use the combining method in the proof of Theorem \ref{thm:big} to amplify the \emph{error} 
in the resulting lower bound. 

From a technical perspective, the primary novelty in the proof of Theorem \ref{thm:big} lies in our choice of an appropriate (and simple) dual witness $\Psi$
for $\OR_t$, and the subsequent analysis
of the correlation of the combined witness $\zeta$ with $\OR_t(f, \dots, f)$. By our choice of $\Psi$, we are able to show that 
 $\zeta$ has correlation with $\OR_t(f, \dots, f)$ that is \emph{exponentially} close to $1$, yielding a lower bound even on the degree of approximations with very high error.

\subsection{Lower Bounds For \acz} 
\label{sec:introacz}

\subsubsection{A New One-Sided Approximate Degree Lower Bound for \acz}
\label{sec:intro-odeg-acz}
Our ultimate goal is to use Theorem \ref{thm:big} to construct a function $F$ in
\acz\ that is hard to approximate by low-degree polynomials even with error exponentially
close to 1. However, in order to apply Theorem \ref{thm:big}, 
we must first identify an \acz\ function $f$ such that $\odeg_{1/2}(f)$ is large.

To this end, we identify fairly general conditions guaranteeing that the one-sided approximate
degree of a function is \emph{equal} to its approximate degree, up to a logarithmic factor. To express our result, let
$[N]=\{1, \dots, N\}$, and let $m, N, R$ be a triple of positive integers such that $R \geq N$, and
$m = N \cdot \log_2 R$. In most cases, we will take $R = N$.
We specifically consider Boolean functions $f$ on $\{-1, 1\}^m$ that interpret their input $x$ as the values of a function $g_x$ mapping $[N] \rightarrow [R]$. That is, we break $x$ up into $N$ blocks each of length $\log_2 R$, and regard each block $x_i$ as the binary representation of $g_x(i)$. Hence, we think of $f$ as computing some \emph{property} $\phi_f$ of
functions $g_x: [N] \rightarrow [R]$. We say that a property $\phi$ is \emph{symmetric} if for all $g: [N] \to [R]$, all permutations $\sigma$ on $[R]$,
and all permutations $\pi$ on $[N]$, it holds that $\phi(g) = \phi(\sigma \circ g \circ \pi)$.
\eat{We say that a function $g: [N] \rightarrow [R]$ satisfies the \emph{small-range} condition if in fact
$g(i) \in [N] \subseteq [R]$
for all $i \in [N]$.}

\begin{restatable}{theorem}{thmsymmetric} \label{thm:symmetric}
Let $f: \{-1, 1\}^m \rightarrow \{-1, 1\}$ be a Boolean function corresponding to a symmetric property $\phi_f$
of functions $g_x : [N] \rightarrow [R]$. 
\eat{Suppose that for all pairs $x, y \in f^{-1}(-1)$ such that $g_x, g_y$ both satisfy the small-range condition, there is a pair of permutations $\sigma, \pi$ such that 
$g_x = \sigma \circ g_y \circ \pi$.}
Suppose that for every pair $x, y \in f^{-1}(-1)$, there is a pair of permutations $\sigma$ on $[R]$ and $\pi$ on $[N]$ such that 
$g_x = \sigma \circ g_y \circ \pi$. Then $\odeg_\eps(f) \ge \frac{1}{\log_2 R} \cdot \adeg_{\eps}(f)$ for all $\eps > 0$. 
\end{restatable}

\paragraph{Proof Idea.}
It is enough to show that any one-sided $\eps$-approximation $p$ to $f$ can be transformed into an actual $\eps$-approximation $r$ to $f$ in a manner that does not increase the degree by too much (i.e., in a manner guaranteeing that $\deg(r) \leq (\log_2 R)\deg(p)$). 

Our transformation from $p$ to $r$ consists of two steps. In the first step,
we turn $p$ into a ``symmetric'' polynomial $p^{\text{sym}}(x) := \mathbb{E}_{y \sim x}[p(y)]$ where $y \sim x$ if $g_y = \sigma \circ g_x \circ \pi$ for some permutations $\sigma$ on $[R]$ and $\pi$ on $[N]$. It follows from work of Ambainis \cite{ambainis} (see Lemma \ref{lem:ambainis-sym}) that the map $p \mapsto p^{\text{sym}}$ increases the degree of $p$ by a factor of at most $\log_2 R$. In the second step,
we argue that there is an affine transformation $r$ of $p^{\text{sym}}$ that is an actual $\eps$-approximation to $f$, completing
the construction.

The existence of the affine transformation $r$ of $p^{\text{sym}}$ follows from two observations: (1) if $p$ is a one-sided approximation for $f$, then so is $p^{\text{sym}}$
(this holds because $\phi_f$ is symmetric),
and (2) $p^{\text{sym}}$ takes on a constant value $v$ on $f^{-1}(-1)$, i.e., $p^{\text{sym}}(x)=v$ for all $x \in f^{-1}(-1)$ (this holds because $x \sim y$ for every pair of inputs $x, y \in f^{-1}(-1)$). These observations imply that even if $p^{\text{sym}}$ is a very poor approximation to $f$ on $f^{-1}(-1)$, we can still obtain a good pointwise approximation $r$ by applying an affine transformation to the range of $p^{\text{sym}}$ that maps $v$ to $-1$ and moves all values closer to $1$.
Section \ref{sec:symmetric} contains the details.

\eat{For all functions $f$ satisfying the conditions of the Theorem \ref{thm:symmetric},
we use a symmetrization argument to generically transform a dual witness $\psi$ for the high approximate degree of $f$
into a dual witness $\psi'$ with one-sided error. Hence, $\psi'$ is in fact a dual witness
for the high \emph{one-sided} approximate degree of $f$. 

Our proof crucially exploits a result of Ambainis \cite{ambainis},
who shows that symmetric properties $\phi_f$ of functions $[N] \rightarrow [R]$, 
the approximate degree of $f$ is unchanged even if $f$ is restricted to inputs satisfying the small-range property. That is, for symmetric properties,
the ``small-range'' case, with $R=N$, is just as hard to approximate by low-degree polynomials as is the large-range case. 
Further details are deferred to Section \ref{sec:symmetric}.}

\medskip

In our primary application of Theorem \ref{thm:symmetric}, we let $f: \{-1, 1\}^m \rightarrow \{-1, 1\}$ be the \ED\ function
(defined in Section \ref{sec:prelims}). Aaronson and Shi \cite{aaronsonshi} showed
that the approximate degree of \ED\ is $\Omega( \left(m/\log m\right)^{2/3})$. \ED\ is computed by a CNF of polynomial size,
and Aaronson and Shi's result remains essentially the best-known lower bound for the approximate degree of a function in \acz. Theorem \ref{thm:symmetric} applies to \ED, yielding the following corollary.

\begin{restatable}{corollary}{cored} \label{cor:ed}
Let $f :\{-1, 1\}^m \rightarrow \{-1, 1\}$ denote the \ED\ function. Then
$\odeg(f) = \tilde{\Omega}(m^{2/3})$. 
\end{restatable}

The best known lower bound on the one-sided approximate degree of an \acz\ function that followed from prior work was $\Omega(m^{1/2})$ (which
holds for the $\AND$ function \cite{NS94, gavinskysherstov}, cf. Fact \ref{fact:odegand} in Section \ref{sec:prelims}). 
Section \ref{sec:discussion} describes some further implications of Theorem \ref{thm:symmetric}.

\medskip
\noindent \textbf{Remark:} In an earlier version of this work, we gave a different \emph{dual} proof of Corollary \ref{cor:ed}. Specifically we showed (cf. Appendix \ref{app:as}) that any dual witness for the high approximate degree of \ED\ can be transformed into a dual witness with one-sided error. This proof in fact shows that $\odeg(f) = \adeg(f)$ (i.e. without incurring the loss of a $1/\log_2 R$ factor as in Theorem \ref{thm:symmetric}). However it remains unclear how to generalize this dual argument to the more general class of properties to which
 Theorem \ref{thm:symmetric} applies (including the \tto\ property discussed in Section \ref{sec:discussion} below). 
 Theorem \ref{thm:symmetric} therefore provides an example of a setting in which the primal view of one-sided approximate degree introduced in this work may be easier to reason about than the dual formulation
used in prior work.

\subsubsection{Accuracy-Degree Tradeoff Lower Bounds for \acz}
By Corollary \ref{cor:ed}, we can apply Theorem \ref{thm:big} to \ED\ to obtain a depth-three 
Boolean circuit $F$ with $t \cdot m$ inputs such that $\adeg_{\eps}(F) = \tilde{\Omega}(m^{2/3})$, for $\eps = 1-2^{-t}$. 
By choosing $t$ and $m$ appropriately, we obtain a depth-three
circuit on $n = t \cdot m$ variables of size $\poly(n)$ 
such that any degree-$d$ polynomial cannot pointwise approximate $F$ to error better than $1-\exp\left(-\tilde{\Omega}(nd^{-3/2})\right)$.

\begin{restatable}{corollary}{coredamp} \label{cor:ed-amp}
For every $d > 0$, there is a depth-3 Boolean circuit $F: \bits^n \to \bits$ of size $\poly(n)$
such that any degree-$d$ polynomial cannot pointwise approximate $F$ to error better than $1-\exp\left(-\tilde{\Omega}(nd^{-3/2})\right)$. In particular, there is a depth-3 circuit $F$ such that any polynomial of degree at most $n^{2/5}$
cannot pointwise approximate $F$ to error better than
$1-\exp\left(-\tilde{\Omega}(n^{2/5})\right)$.
\end{restatable}

\eat{
In a previous version of this work, we conjectured that this circuit $F$ has \emph{threshold degree} $\tilde{\Omega}(n^{2/5})$. That is, any polynomial that simply sign-represents $F$ requires degree  $\tilde{\Omega}(n^{2/5})$. This conjecture was subsequently resolved in the affirmative by Sherstov \cite{sherstovmp} (cf. Section \ref{sec:sherstov}).
}

\subsubsection{Discrepancy Upper Bound}
Discrepancy, defined formally in Section \ref{sec:discrepancy}, is a central quantity in communication
complexity and circuit complexity. For instance, upper bounds on the discrepancy 
of a function $f$ immediately
yield lower bounds on the cost of small-bias communication protocols for computing $f$ (Section \ref{sec:applications} has details).
\begin{table}
\centering
\begin{tabular}{|c|c|c|}
\hline
Reference & Discrepancy Bound &  Circuit Depth\\
\hline
Sherstov \cite{majmaj}& $\exp(-\Omega(n^{1/5}))$ & 3\\
\hline
Buhrman et al. \cite{bvdw} & $\exp(-\Omega(n^{1/3}))$  & 3\\
\hline
Sherstov \cite{patternmatrix} & $\exp(-\Omega(n^{1/3})) $ & 3\\
\hline
This work & $\exp\left(-\tilde{\Omega}(n^{2/5})\right)$  & 4\\
\hline
\end{tabular}
\caption{Comparison of our new discrepancy bound for \acz\ to prior work. The circuit depth column lists the depth of the circuit used to exhibit the  bound.}
\end{table}
The first exponentially small discrepancy upper bounds for \acz\ were proved by Burhman \etal\ \cite{bvdw} and Sherstov \cite{majmaj, patternmatrix},
who exhibited constant-depth circuits with discrepancy $\exp(-\Omega(n^{1/3}))$. Our results improve the best-known upper bound
to $\exp\left(-\tilde{\Omega}(n^{2/5})\right)$. 

Our result relies on a powerful technique developed by Sherstov \cite{patternmatrix}, known as the pattern-matrix method (stated as Theorem \ref{thm:patternmatrix} in Section \ref{sec:discrepancy}). This technique allows one to automatically translate lower bounds on 
the $\eps$-approximate degree of a Boolean function $F$
into lower bounds on the \emph{discrepancy} of a related function $F'$
as long as $\eps$ is exponentially close to one.
By applying the pattern-matrix method to Corollary \ref{cor:ed-amp},
we obtain the following result.

\begin{restatable}{corollary}{cordisc} \label{cor:discrepancy}
There is a depth-4 Boolean circuit $F': \bits^n \to \bits$ with
discrepancy $\exp\left(-\tilde{\Omega}(n^{2/5})\right)$.
\end{restatable}

\subsubsection{Threshold Weight Lower Bound}
A \emph{polynomial threshold function} (PTF) for a Boolean function $f$ is a multilinear polynomial $p$ with integer coefficients that agrees in sign with $f$ on all Boolean inputs. The \emph{weight} of an
$n$-variate polynomial $p$ is the sum of the absolute
value of its coefficients. 
The \emph{degree-$d$ threshold weight} of a Boolean
function $f: \{-1, 1\}^n \rightarrow \{-1, 1\}$,
denoted $W(f, d)$, refers to the least weight
of a degree-$d$ PTF for $f$.
We let $W(f)$ denote the quantity $W(f, n)$, i.e.,
the least weight of any threshold function for $f$ regardless of 
its degree. As discussed in Section \ref{sec:applications}, threshold weight
has important applications in learning theory.

Threshold weight is closely related to $\eps$-approximate degree when $\eps$ is very close to $1$ (see Lemma \ref{lem:weight-relations} in \secref{sec:weightrelations}). This allows us to translate Corollary \ref{cor:ed-amp}
into a lower bound on the degree-$d$ threshold weight of \acz.

\begin{restatable}{corollary}{coredampthresh} \label{cor:ed-ampthresh}
For every $d > 0$, there is a depth-3 Boolean circuit $F: \bits^n \to \bits$ of size $\poly(n)$
such that $W(F, d) \geq \exp\left(\tilde{\Omega}(nd^{-3/2})\right)$.  In particular, 
$W(F, n^{2/5}) = \exp\left(\tilde{\Omega}(n^{2/5})\right)$.
\end{restatable}

A result of Krause \cite{krause} (see Lemma \ref{lem:threshold-weight} in Section \ref{sec:ptfweight}) allows us to extend our new degree-$d$ threshold weight lower bound for $F$ into a \emph{degree independent} threshold weight lower bound for a related function $F'$ (we also provide a new and simple proof of Krause's result based on LP duality, cf. Appendix \ref{app:threshold-weight}). The previous best lower bound on the threshold weight of \acz\
was $\exp\left(\Omega(n^{1/3})\right)$, due to 
Krause and Pudl\'{a}k \cite{krausepudlak}.

\begin{restatable}{corollary}{corweightacz} \label{cor:weightacz}
There is a depth-4 Boolean circuit $F': \bits^n \to \bits$ 
satisfying $W(F') = \exp\left(\tilde{\Omega}(n^{2/5})\right)$. 
\end{restatable}

Moreover, while the threshold weight bound 
of Corollary \ref{cor:weightacz} is stated for polynomial
threshold functions over $\{-1, 1\}^n$, we 
show that the same threshold weight
lower bound also holds for polynomials over $\{0, 1\}^n$.

\subsection{Approximate Degree Lower Bounds for AND-OR Trees}
\label{sec:introANDOR}
The $d$-level AND-OR tree on $n$ variables is a function described by a read-once circuit
of depth $d$ consisting of alternating layers of AND gates and OR gates.
We assume for simplicity that all gates have fan-in $n^{1/d}$. For example, the two-level
AND-OR tree is a read-once CNF in which all gates have fan-in $n^{1/2}$.

Until recently, the approximate degree of AND-OR trees of depth two or greater had resisted characterization, despite 
19 years of attention \cite{NS94, hoyer, shi, ambainis, sherstovFOCS, sherstovnew, bunthaler}. The case of
of depth two was reposed as a challenge problem by Aaronson in 2008 \cite{scottslides}, as it captured
the limitations of existing lower bound techniques. This case was resolved last year by the current authors \cite{bunthaler}, and independently by Sherstov \cite{sherstovnew}, who proved a lower bound of $\Omega(\sqrt{n})$, matching
an upper bound of H\o yer, Mosca, and de Wolf \cite{hoyer}. However, the case of depth
three or greater remained open. To our knowledge, the best known lower bound for $d \geq 3$
was $\Omega(n^{1/4+1/2d})$, which follows by combining the depth-two lower bound \cite{sherstovnew, bunthaler} with an earlier direct-sum theorem of Sherstov \cite[Theorem 3.1]{sherstovFOCS}. 

By combining the techniques of our earlier work \cite{bunthaler} with our hardness amplification result (Theorem \ref{thm:big}),
we improve this lower bound to $\Omega\left(n^{1/2}/\log^{(d-2)/2}(n)\right)$ for any constant $d\geq2$. 
A line of work on quantum query algorithms \cite{hoyer, spalekandor, reichardt} 
established an upper bound of $O(n^{1/2})$ for AND-OR trees of any depth,
demonstrating that our
result is optimal up to polylogarithmic factors (see Section \ref{sec:andortree} for details). 

\begin{restatable}{theorem}{thmandor} \label{thm:andor} Let $\ANDOR_{d, n}$ denote the $d$-level AND-OR tree on $n$ variables. 
Then $\adeg(\ANDOR_{d, n}) = \Omega\left(n^{1/2}/\log^{(d-2)/2}n\right)$ for any
constant $d>0$.
\end{restatable}

\paragraph{Proof Idea.}
To introduce our proof technique, we first describe the method used in \cite{bunthaler} to
construct an optimal dual polynomial in the case $d=2$, and we identify why this method breaks down 
when trying to extend to the case $d=3$.
We then explain how to use our hardness amplification result (Theorem \ref{thm:big})
to construct a different dual polynomial that does extend to the case $d=3$. 

Let $M$ denote the fan-in of all gates in $\ORAND_{2, M^{2}}$. 
In our earlier work \cite{bunthaler}, we constructed a dual polynomial for $\ORAND_{2, M^2}$ as follows.\footnote{We actually constructed a dual polynomial for $\ANDOR_{2,M^2}$, but the analysis for the case of $\ORAND_{2, M^2}$ is entirely analogous.}
By Fact \ref{fact:odegand} there is a dual polynomial $\gamma_1$ witnessing the fact that $\odeg(\AND_{M})=\Omega\left(M^{1/2}\right)$, and a dual polynomial $\gamma_2$ witnessing the fact that $\adeg(\OR_M)=\Omega\left(M^{1/2}\right)$.
We then combined the dual witnesses $\gamma_1$ and $\gamma_2$, using 
the same ``combining'' technique as in the proof of Theorem \ref{thm:big}, 
to obtain a dual witness $\gamma_3:\{-1,1\}^{M^2}\rightarrow \mathbb{R}$ for the high approximate degree
of $\ORAND_{2, M^2}$.

Recall that we say a dual witness has \emph{pure high degree $d$} if it has zero correlation
with every polynomial of degree at most $d$.  
It followed from earlier work \cite{sherstovFOCS} that $\gamma_3$ has pure high degree equal to the product
of the pure high degree of $\gamma_1$ and the pure high degree of $\gamma_2$, 
yielding an $\Omega(M)$ lower bound on the pure high degree of $\gamma_3$. 
The new ingredient of the analysis in \cite{bunthaler} was to use the one-sided error 
of the ``inner'' dual witness $\gamma_1$ to argue that $\gamma_3$ also had good correlation with $\ORAND_{2, M^2}$. 

\medskip
\noindent \textbf{Extending to Depth Three.}
Let $M=n^{1/3}$ denote the fan-in of all gates in $\ANDOR_{3, n}$. 
In constructing a dual witness for $\ANDOR_{3, n}=\AND_{M}(\ORAND_{2, M^2}, \dots, \ORAND_{2, M^2})$, it is natural to try the following approach. 
Let $\gamma_4$ be a dual polynomial witnessing the fact that the approximate degree
of $\AND_{M}=\Omega(\sqrt{M})$. Then we can combine $\gamma_3$ and $\gamma_4$ in the same manner as above
to obtain a dual function $\gamma_5$.

The difficulty in establishing that $\gamma_5$ is a dual witness to the high approximate degree of $\ANDOR_{3, n}$
is in showing that $\gamma_5$ has good correlation with $\ANDOR_3$. In our earlier work,
we showed $\gamma_3$ has large correlation with $\ORAND_{2, n}$ by exploiting the fact that the inner dual witness
$\gamma_1$ had one-sided error, i.e., $\gamma_1(y)$ agrees in sign with $\AND_{M}$ whenever $y \in \AND^{-1}_M(-1)$ .
 However, $\gamma_3$ itself does not satisfy an analogous 
 property: there are inputs $x_i \in \ORAND^{-1}_{2, M^2}(-1)$ such that $\gamma_3(x_i) > 0$,
 \emph{and} there are inputs $x_i \in \ORAND^{-1}_{2, M^2}(1)$ such that $\gamma_3(x_i) < 0$.
 
 To circumvent this issue, we use a different inner dual witness $\gamma'_3$ in place of $\gamma_3$.
 Our construction of $\gamma'_3$ utilizes our hardness amplification analysis to achieve the following:
 while $\gamma'_3$ has error ``on both sides'', the error from the ``wrong side'' is very small. 
 The hardness amplification step causes $\gamma'_3$ to have pure high degree that is lower than that of
 the dual witness $\gamma_3$ constructed in \cite{bunthaler} by a $\sqrt{\log n}$ factor. 
 However, the hardness amplification step permits us to prove the desired lower bound on the correlation of $\gamma_5$ with $\ANDOR_{3, n}$. The proof for the general case, which is quite technical, can be found in Section \ref{sec:andortree}.

\subsection{Lower Bounds for Read-Once DNFs and CNFs}
\label{sec:readonceintro}
Our techniques also yield new lower bounds 
on the approximate degree and degree-$d$ threshold weight of read-once DNF and CNF formulas.
Before stating our results, we discuss relevant prior work.

In their seminal work on perceptrons, Minsky and Papert
exhibited a read-once DNF $f:\bits^n \rightarrow \bits$ with \emph{threshold degree}
$\Omega(n^{1/3})$ \cite{mp}. That is, a real polynomial requires degree $\Omega(n^{1/3})$ just to agree with $f$ in sign. However, to our knowledge
no non-trivial lower bound on the degree-$d$ threshold
\emph{weight} of read-once DNFs was known for any
$d = \omega(n^{1/3})$. 

In an influential result, Beigel \cite{beigelomb} exhibited a polynomial-size (read-many) DNF
called \omb\ satisfying the following: there is some constant $\delta > 0$ such that $\adeg_{1-2^{-\delta n/d^2}}(\ombmath) > d$, and hence also $W(\ombmath, d) = \exp\left(\Omega(n/d^2)\right)$ (see Section \ref{sec:weightrelations}). 
Motivated by applications in computational learning theory (see Section \ref{sec:applications}), Klivans and Servedio showed that Beigel's lower bound
is essentially tight for $d < n^{1/3}$ \cite{ksomb}.
Very recently, Servedio, Tan, and Thaler showed an alternative
lower bound on the degree-$d$ threshold weight of \omb.
Specifically, they showed that
$W(\ombmath, d) = \exp\left(\Omega\left(\sqrt{n/d}\right)\right)$ \cite{thalercolt}.
The lower bound of Servedio et al. improves over Beigel's for any $d > n^{1/3}$, and is essentially
tight in this regime (i.e., when $d > n^{1/3}$). 

While \omb\ is a relatively simple DNF (in fact, it is 
a \emph{decision list}), it is not a read-once DNF. 
Our results extend the lower bounds of Servedio et al. and Beigel from decision lists to read-once DNFs and CNFs. In the statement
of the results below, we restrict ourselves to DNFs, as the case of CNFs
is entirely analogous.  

\subsubsection{Extending the Lower Bound of Servedio et al. to Read-Once DNFs} 
In order to extend the lower bound of Servedio et al. to read-once
DNFs and CNFs,
we extend our hardness amplification techniques
from one-sided approximate degree
to a new quantity we call \emph{degree-$d$ one-sided non-constant approximate weight}.
This quantity captures the least $L_1$ \emph{weight} (excluding the constant term) of a polynomial 
of degree at most $d$ that is a one-sided approximation of $f$.
We denote the degree-$d$ one-sided approximate weight
of a Boolean function $f$ by $W^*_{\eps}(f, d)$, where $\eps$
is an error parameter. 

We prove the following analog of Theorem \ref{thm:big}.
\begin{theorem} \label{thm:big-weights}
Fix $d > 0$. Let $f: \bits^m \rightarrow \{-1, 1\}$, and suppose that $W^*_{3/4}(f, d) > w$. Let $F: \bits^{m \cdot t} \to \bits$ denote the function $\OR_t(f, \dots, f)$. Then any degree-$d$ polynomial that approximates $F$ to within error $1-2^{-t}$ requires weight $2^{-5t}w$.
\end{theorem}

Adapting a proof of Servedio et al., we can show that
$W^*_{3/4}(\AND_m, d) \geq 2^{\Omega(m/d)}$. By applying Theorem
\ref{thm:big-weights} with $f = \AND_m$, along
with standard manipulations, we are able to extend
the lower bound of Servedio et al. to read-once CNFs and DNFs.

\begin{restatable}{corollary}{corcolt} \label{cor:colt}
For each $d=o(n/\log^4 n)$, there is a read-once DNF $F$
satisfying $W(F, d) = \exp\left(\Omega\left(\sqrt{n/d}\right)\right)$. 
\end{restatable}

In particular, there is a read-once DNF that cannot be computed by any PTF of $\poly(n)$ weight, unless the degree is $\tilde{\Omega}(n)$.

\subsubsection{Extending Beigel's Lower Bound to Read-Once DNFs}
It is known that $\odeg(\AND_m) = \Omega(m^{1/2})$ (cf. Fact \ref{fact:odegand}). 
By applying Theorem \ref{thm:big} with $f=\AND_m$,
we obtain the following result.

\begin{restatable}{corollary}{corDNFadeg} \label{cor:DNF-adeg}
There is an (explicit) read-once DNF $F: \{-1, 1\}^n \to \{-1, 1\}$ with $\adeg_{1-2^{-n/d^2}}(F) = \Omega(d)$.
\end{restatable}

We remark that for $d < n^{1/3}$, 
Corollary \ref{cor:DNF-adeg} is subsumed
by Minsky and Papert's seminal result that exhibited a read-once
DNF $F$ with threshold degree $\Omega(n^{1/3})$
\cite{mp}.
However, for $d > n^{1/3}$, it is not subsumed
by Minsky and Papert's result, nor by
Corollary \ref{cor:colt}. Indeed, 
Corollary \ref{cor:colt} yields a lower bound on the degree-$d$
threshold weight of read-once DNFs, but not a lower bound
on the \emph{approximate-degree} of read-once DNFs (see Section \ref{sec:weightrelations} for further discussion on the separation between these quantities).

\subsection{Discussion} \label{sec:discussion}

\subsubsection{Subsequent Work by Sherstov}
In 1969, Minsky and Papert gave a lower bound of $\Omega(n^{1/3})$ on the threshold degree of an explicit read-once DNF formula.
Klivans and Servedio \cite{ks} proved their lower bound to be tight within a logarithmic factor for DNFs of polynomial size,
but it remained a well-known open question to give a threshold degree lower bound of $\Omega(n^{1/3+ \delta})$ for a function in \acz; the only progress prior to our work was due to O'Donnell and Servedio \cite{os}, who established an $\Omega(n^{1/3}\log^k n)$ lower bound for any constant $k>0$.

Let $f$ denote the \ED\ function on $n^{3/5}$ variables. In an earlier version of this work, we conjectured that the function
 $F=\OR_{n^{2/5}}(f, \dots, f)$
appearing in Corollary \ref{cor:ed-amp}
in fact satisfies $\deg_\pm(f) = \tilde{\Omega}(n^{2/5})$, and observed that this would yield the first
polynomial improvement on Minsky and Papert's lower bound. 
Sherstov \cite[Theorem 7.1]{sherstovmp} has recently proved our conjecture. His proof, short and elegant, extends our dual witness construction
in the proof of Theorem \ref{thm:big} to establish
a different form of hardness amplification, from one-sided approximate degree to threshold degree.
Specifically, he shows that if a Boolean function $f$ has one-sided approximate degree $d$, then the block-wise composition $\OR_t(f, \dots, f)$ has threshold degree at least $\min\{ct, d\}$ for some constant $c$. This result is incomparable to our Theorem \ref{thm:big} when $t \leq d$, but when $t \gg d$, Sherstov's result is a substantial strengthening of Theorem \ref{thm:big}.

In the same work, Sherstov has also proven a much stronger and more difficult result: for any $k > 2$, he gives a read-once formula 
of depth $k$ with threshold degree $\Omega\left(n^{(k-1)/(2k-1)}\right)$. Notice that for any constant $\delta > 0$, this yields an \acz\ function
with threshold degree $\Omega(n^{1/2-\delta})$. This in turn yields an improvement of our discrepancy upper bound (Corollary \ref{cor:discrepancy}) for \acz\ to $\exp(-\Omega(n^{1/2-\delta}))$, and of our threshold weight lower bound (Corollary \ref{cor:weightacz}) to $\exp(\Omega(n^{1/2-\delta}))$.

\subsubsection{Subsequent Work by Kanade and Thaler}
Existing applications of one-sided approximate degree \cite{gavinskysherstov, sherstovnew, bunthaler, sherstovmp} have all been of
a negative nature (proving communication and circuit lower bounds, establishing limitations on existing PAC learning algorithms, etc.). 
 Kanade and Thaler \cite{kanadethaler} have recently identified a positive (algorithmic) application of one-sided approximate degree.
Specifically, they show that one-sided approximate degree upper bounds imply fast algorithms 
 in the reliable agnostic learning framework of Kalai et al. \cite{reliable}. This framework captures learning tasks in which one type of error (such as false negative errors) is costlier than other types. Kanade and Thaler use this result to give the first
 sub-exponential time algorithms for distribution-independent reliable learning of several fundamental concept classes.
 
 \medskip
 
In light of these developments, we are optimistic that the notion of one-sided approximate degree will continue to enable progress on questions within the analysis of Boolean functions and computational complexity theory.

\subsubsection{Future Directions}
Beame and Machmouchi \cite{beame} established an $\Omega(n/\log n)$ lower bound on the quantum query complexity
of a specific function $f: \{-1, 1\}^n \rightarrow \{-1, 1\}$ in \acz. The previous best lower bound was $\Omega((n/\log n)^{2/3})$,
which held for the \ED\ function \cite{aaronsonshi}.

Beame and Machmouchi's lower bound applies to the \tto\ function, which is computed
by depth-three circuit of polynomial size. This function interprets its input
as a list of $N$ numbers from a range of size $R \geq N$, and evaluates to $-1$ if and only if exactly $N/2$ numbers appear 
in the list, each with frequency exactly 2. They pose as an open question the problem of
resolving the approximate degree of \tto\footnote{Technically speaking, they ask about the \ttov\ function, which is 
a promise variant of the \tto\ function.}  (recall that the approximate degree of $f$ is a lower bound on the quantum query complexity of $f$,
but polynomial separations between approximate degree and quantum query complexity are known \cite{ambainissep}). 

For simplicity, we focus on the case where $R = N$. We observe that Theorem \ref{thm:symmetric} applies to the \tto\ function, revealing that its one-sided approximate
degree is almost equal to its approximate degree. 

\begin{corollary} \label{cor:conclusioncor}
Let $f: \{-1, 1\}^m \rightarrow \{-1, 1\}$ denote the \tto\ function on $m$ variables. For any $\eps> 0$, $\odeg_{\eps}(f) \ge \adeg_{\eps}(f) / \log m$.
\end{corollary}

Combining Corollary \ref{cor:conclusioncor} and the recent result \cite[Theorem 7.1]{sherstovmp} 
allows us to transform any \emph{approximate degree} lower bound for the \tto\ function into a \emph{threshold degree}
lower bound for a related depth-four circuit.

\begin{corollary}
Let $f: \{-1, 1\}^m \rightarrow \{-1, 1\}$ denote the \tto\ function on $m$ variables, and let $d=\odeg(f) \ge \adeg(f) / \log m$.
Let $n=m \cdot d$, and define $F: \{-1, 1\}^{n} \rightarrow \{-1, 1\}$ via
$F=\OR_d(f, \dots, f)$. Then $\deg_{\pm}(F)=\Omega(d)$.  
In particular, if $\adeg(f)=\Omega(m/\log m)$, then $\deg_\pm(F) = \Omega\left(n^{1/2}/\log n\right)$.
\end{corollary}

Thus, establishing a quasilinear lower bound on the approximate degree of \tto\ would immediately
yield a function $F$ computable by a depth-four circuit of polynomial size with threshold degree $\tilde{\Omega}(n^{1/2})$,
a polynomial improvement over Sherstov's $\Omega(n^{(k-1)/(2k-1)})$ bound for any constant depth $k$. 
Even a lower bound of $\Omega(m^{3/4+\delta})$ for some positive constant $\delta$ on the approximate degree of the \tto\ function would yield a depth four
circuit with threshold degree $\Omega(n^{3/7 + \delta'})$ for some $\delta'>0$. This would constitute a polynomial improvement
over the current best lower bound of $\Omega(n^{3/7})$ for depth 4, and would additionally imply improved lower bounds on the threshold weight and discrepancy
of depth five circuits.

While the best known lower bound on the approximate degree of the \tto\ function on $m$ variables is currently $\Omega((m/\log m)^{2/3})$ (this can 
be derived by reduction to \ED), we conjecture that its approximate degree is in fact $\Omega(m/\log m)$,
and interpret Beame and Machmouchi's quantum query lower bound as providing mild evidence in favor of this conjecture.

\subsection{Paper Roadmap}
Section \ref{sec:prelims} establishes terminology, introduces our main technique based on LP-duality, and proves essential technical lemmas.
Section \ref{sec:highlow} establishes
our central hardness amplification result for one-sided approximate
degree (Theorem \ref{thm:big}). 
Section \ref{sec:aczlb} establishes our new one-sided approximate degree lower bound 
for \acz\ (Theorem \ref{thm:symmetric}, Corollary \ref{cor:ed}), and combines this with Theorem \ref{thm:big}
to obtain our new
lower bounds on ``accuracy vs. degree'' tradeoffs for pointwise
approximating \acz\ by polynomials (Corollary \ref{cor:ed-amp}).
It then proves our new discrepancy
upper bound for an \acz\ function (Corollary \ref{cor:discrepancy}) and our new 
threshold weight lower bound for \acz\ (Corollaries \ref{cor:ed-ampthresh} and \ref{cor:weightacz}). Section \ref{sec:andortree} proves our new approximate degree lower bound for AND-OR trees
(Theorem \ref{thm:andor}).
Section \ref{sec:readonce} proves our new lower bounds
for read-once DNFs (Theorem \ref{thm:big-weights}, Corollary \ref{cor:colt}, and Corollary \ref{cor:DNF-adeg}).
Section \ref{sec:applications} highlights several applications of these results to communication complexity, circuit lower bounds, and learning theory. 

\section{Preliminaries}
\label{sec:prelims}
We work with Boolean functions $f: \{-1, 1\}^n \to \{-1, 1\}$
under the standard convention that 1 corresponds to logical false, and $-1$ corresponds to logical true. For a real-valued function $r : \bits^n \to \R$, we let $\|r\|_\infty = \max_{x \in \bits^n} |r(x)|$ denote the $\ell_\infty$ norm of $r$. 
We let OR$_n$ and AND$_n$ denote the OR function and AND function on $n$ variables respectively. Define $\sgn(t) = -1$ if $t \le 0$ and 1 otherwise. For a set $S \subseteq [n] = \{1, \dots, n\}$, let $\chi_S(x) := \prod_{i \in S} x_i$ denote the parity function
over variables indexed by $S$.

We now define the notions of approximate degree, approximate weight,
threshold degree, threshold weight, and their one-sided variants.
 
\subsection{Polynomial Approximations and their Dual Characterizations}
\label{sec:duallps}
\subsubsection{Approximate Degree}
\label{sec:approxdegree}
The $\eps$-approximate degree of a function $f: \{-1, 1\}^n \rightarrow \{-1, 1\}$,
denoted $\adeg_\eps(f)$, is the minimum (total) degree of any real polynomial $p$ such that $\|p - f\|_\infty \le \eps$, i.e., $|p(x) - f(x)| \leq \eps$ for all $x \in \{-1, 1\}^n$.
We use $\adeg(f)$ to denote $\adeg_{1/3}(f)$, and use this to refer to the \emph{approximate degree} of a function without qualification. The choice of $1/3$ is arbitrary, as $\widetilde{\deg}(f)$ is related to $\adeg_\eps(f)$ by a constant factor for any constant $\eps \in (0, 1)$. 

Given a Boolean function $f$, let $p$ be a real polynomial that minimizes $\|p - f\|_\infty$ among
all polynomials of degree at most $d$. Since we work over $x \in \{-1, 1\}^n$, we may assume without loss of generality that $p$ is multilinear with the representation $p(x) = \sum_{|S| \leq d} c_S \chi_S(x)$ where the coefficients $c_S$ are real numbers. Then $p$ is an optimum of the following linear program.

\[ \boxed{\begin{array}{lll} 
    \text{min}  &     \eps    \\
    \mbox{such that} &\Big|f(x) - \sum_{|S| \leq d} c_S \chi_S(x)\Big| \leq \eps & \text{ for each } x \in \{-1, 1\}^n\\
    &c_S \in \mathbb{R} & \text{ for each } |S| \leq d\\
    &\eps \geq 0
    \end{array}}
\]

The dual LP is as follows.

\[ \boxed{\begin{array}{lll} 
    \text{max} &    \sum_{x \in \{-1, 1\}^n} \phi(x) f(x)   \\
    \mbox{such that} &\sum_{x \in \{-1, 1\}^n} |\phi(x)| = 1\\
    &\sum_{x \in \{-1, 1\}^n} \phi(x) \chi_S(x)=0  & \text{ for each } |S| \leq d\\
    &\phi(x) \in \mathbb{R} & \text{ for each } x \in \{-1, 1\}^n
    \end{array}}
\]

Strong LP-duality thus yields the following well-known dual characterization of approximate degree (cf. \cite{patternmatrix}).

\begin{theorem} \label{thm:prelim} Let $f: \{-1, 1\}^n \to \{-1, 1\}$ be a Boolean function. Then $\adeg_\eps(f) > d$ if and only if there is a polynomial $\phi: \{-1, 1\}^n \rightarrow \mathbb{R}$ such that 
\begin{equation} \label{eq:prelim0} \sum_{x \in \{-1, 1\}^n} f(x) \phi(x) > \eps, \end{equation}
\begin{equation} \label{eq:prelim1} \sum_{x \in \{-1, 1\}^n} |\phi(x)| = 1,\end{equation}  and
\begin{equation} \label{eq:prelim2} \sum_{x \in \{-1, 1\}^n} \phi(x) \chi_S(x)=0   \text{ for each } |S| \leq d.\end{equation}
\end{theorem}

If $\phi$ satisfies \eqref{eq:prelim2}, we say $\phi$ has \emph{pure high degree} $d$.
We refer to any feasible solution $\phi$ to the dual LP as a \emph{dual polynomial} for $f$.

\subsubsection{One-Sided Approximate Degree} \label{sec:odeg}
We introduce a relaxed notion of the approximate degree of $f$ which we call the one-sided $\eps$-approximate degree, denoted by $\odeg_\eps(f)$. This is the least degree of a real polynomial $p$ with that is an $\eps$-\emph{one-sided approximation} to $f$, meaning

\begin{enumerate}
\item $|p(x) - 1| \leq \eps$ for all $x \in f^{-1}(1)$.
\item $p(x) \leq -1+\eps$ for all $x \in f^{-1}(-1)$.
\end{enumerate}

That is, we require $p$ to be very accurate on inputs in $f^{-1}(1)$, but only require ``one-sided accuracy'' on inputs in  $f^{-1}(-1)$.
We use $\odeg(f)$ to denote $\odeg_{1/3}(f)$, and refer to this quantity without qualification as the \emph{one-sided approximate degree} of $f$. 

The primal and dual LPs change in a simple but crucial way if we look at one-sided approximate degree rather than approximate degree. 
Let  $p(x) = \sum_{|S| \leq d} c_S \chi_S(x)$ be a polynomial of degree $d$ for which the $\eps$-one-sided approximate degree of $f$ is attained. 
Then $p$ is an optimum of the following linear program.

\[ \boxed{\begin{array}{lll} 
    \text{min}  &     \eps    \\
    \mbox{such that} &\Big|f(x) - \sum_{|S| \leq d} c_S \chi_S(x)\Big| \leq \eps & \text{ for each } x \in f^{-1}(1)\\
    &\sum_{|S| \leq d} c_S \chi_S(x) \leq -1+\eps & \text{ for each } x \in f^{-1}(-1)\\
    &c_S \in \mathbb{R} & \text{ for each } |S| \leq d\\
    &\eps \geq 0 
    \end{array}}
\]

The dual LP is as follows.

\[ \boxed{\begin{array}{lll} 
    \text{max} &    \sum_{x \in \{-1, 1\}^n} \phi(x) f(x)   \\
    \mbox{such that} &\sum_{x \in \{-1, 1\}^n} |\phi(x)| = 1\\
    &\sum_{x \in \{-1, 1\}^n} \phi(x) \chi_S(x)=0  & \text{ for each } |S| \leq d\\
    & \phi(x) \leq 0 \text{ for each } x \in f^{-1}(-1)\\
    &\phi(x) \in \mathbb{R} & \text{ for each } x \in \{-1, 1\}^n
    \end{array}}
\]

We again appeal to strong LP-duality for the following dual characterization of one-sided approximate degree.

\begin{theorem} \label{thm:oprelim} Let $f: \{-1, 1\}^n \to \{-1, 1\}$ be a Boolean function. Then $\odeg_\eps(f) > d$ if and only if there is a polynomial $\phi: \{-1, 1\}^n \rightarrow \mathbb{R}$ such that 
\begin{equation} \label{eq:oprelim0} \sum_{x \in \{-1, 1\}^n} f(x) \phi(x) > \eps, \end{equation}
\begin{equation} \label{eq:oprelim1} \sum_{x \in \{-1, 1\}^n} |\phi(x)| = 1,\end{equation} 
\begin{equation} \label{eq:oprelim2} \sum_{x \in \{-1, 1\}^n} \phi(x) \chi_S(x)=0   \text{ for each } |S| \leq d,\end{equation} and
\begin{equation} \label{eq:oprelim3} \phi(x) \le 0   \text{ for each } x \in f^{-1}(-1).\end{equation}
\end{theorem}

Observe that a feasible solution $\phi$ to this dual LP is a feasible solution to the dual LP for approximate degree, with the additional constraint that $\phi(x)$ agrees in sign with $f(x)$ whenever $x \in f^{-1}(-1)$. We refer to any such feasible solution $\phi$ as a dual polynomial for $f$ with \emph{one-sided error}.
Dual polynomials with one-sided error have recently played an important role in resolving open problems in communication complexity \cite{gavinskysherstov} and resolving
the approximate degree of the two-level AND-OR tree \cite{sherstovnew, bunthaler}. They will play a critical role in our proof of Theorem \ref{thm:big} as well.

Prior work using the dual formulation of one-sided approximate degree exploited the fact that the AND function has one-sided approximate degree equal to its ordinary approximate degree \cite{gavinskysherstov, sherstovnew, bunthaler}. This fact also plays an important role in the applications of our hardness amplification technique to AND-OR trees (Section \ref{sec:andortree}) and to read-once DNF formulas (Section \ref{sec:readonce}).

\begin{fact} \label{fact:odegand}
$$\odeg(\AND_m) = \adeg(\AND_m) = \Omega(\sqrt{m}).$$
\end{fact}

Fact \ref{fact:odegand} can be seen by observing that Nisan and Szegedy's proof that $\adeg(\AND_m) = \Omega(\sqrt{m})$ in fact extends to one-sided approximate degree \cite{NS94}. Alternatively, it can be directly shown that any dual witness (as defined in Theorem \ref{thm:prelim}) for the fact that $\adeg(\AND_m)=\Omega(\sqrt{m})$ must have one-sided error (cf. \cite[Theorem 5.1]{gavinskysherstov}).

\subsubsection{Approximate Weight}
We define the \emph{degree-$d$ $\eps$-approximate weight} of $f$, $W_\eps(f, d)$, to be the minimum weight of a degree-$d$ polynomial that approximates $f$ pointwise to error $\eps$. Recall that the weight of a polynomial $p$ is the $L_1$ norm of its coefficients.
If $\adeg_\eps(f) > d$, we define $W_\eps(f, d) = \infty$.

For a fixed error parameter $\eps$ and degree $d$, the degree-$d$ $\eps$-approximate weight of a function $f$ is captured by the following optimization problem.

\[ \boxed{\begin{array}{lll} 
    \text{min}  &     \sum_{|S| \le d} |c_S|    \\
    \mbox{such that} &\Big|f(x) - \sum_{|S| \leq d} c_S \chi_S(x)\Big| \leq \eps & \text{ for each } x \in \{-1, 1\}^n\\
    &c_S \in \mathbb{R} & \text{ for each } |S| \leq d\\
    \end{array}}
\]

A standard substitution of each term $|c_S|$ in the objective with an auxiliary non-negative variable $c_S'$, as well as the addition of the constraints $c_S \le c_S'$ and $-c_S \le c_S'$ shows that this is in fact a linear program. The dual LP is as follows.

\[ \boxed{\begin{array}{lll} 
    \text{max} &    \sum_{x \in \{-1, 1\}^n} \phi(x) f(x) -\eps \sum_{x \in \{-1, 1\}^n} |\phi(x)| \\
    \mbox{such that}  & \left|\sum_{x \in \{-1, 1\}^n} \phi(x) \chi_S(x) \right| \le 1  & \text{ for each } |S| \leq d\\
    &\phi(x) \in \mathbb{R} & \text{ for each } x \in \{-1, 1\}^n
    \end{array}}
\]

We thus obtain the following duality theorem.

\begin{theorem} \label{thm:wprelim}
Let $f: \{-1, 1\}^n \to \{-1, 1\}$ be a Boolean function. Then $W_\eps(f, d) > w$ if and only if there is a polynomial $\phi: \{-1, 1\}^n \to \R$ such that
\begin{equation} \label{eq:wprelim0} \sum_{x \in \{-1, 1\}^n} f(x) \phi(x) - \eps\sum_{x \in \{-1, 1\}^n} |\phi(x)| > w, \end{equation}
\begin{equation} \label{eq:wprelim1} \left|\sum_{x \in \{-1, 1\}^n} \phi(x) \chi_S(x)\right| \le 1   \text{ for each } |S| \leq d.\end{equation}
\end{theorem}

\subsubsection{One-Sided Non-Constant Approximate Weight}
To derive our new lower bound on the degree-$d$ threshold weight of read-once DNFs (Corollary \ref{cor:colt}), we need the following technical variation on approximate weight. Given a polynomial $p(x) = \sum_S c_S\chi_S(x)$, define the \emph{non-constant weight} of $p$ to be the $L_1$ norm of its coefficients excluding the constant term, i.e., $\sum_{S \ne \emptyset} |c_S|$. We then define the \emph{degree-$d$ one-sided non-constant $\eps$-approximate weight} of $f$, denoted by $W^*_\eps(f, d)$ to be the minimum non-constant weight of an $\eps$-one-sided approximation to $f$. Linear programming duality
yields the following characterization of $W^*_\eps(f, d)$.

\begin{theorem} \label{thm:owprelim}
Let $f: \{-1, 1\}^n \to \{-1, 1\}$ be a Boolean function. Then $W^*_\eps(f, d) > w$ if and only if there is a polynomial $\phi: \{-1, 1\}^n \to \R$ such that
\begin{equation} \label{eq:owprelim0} \sum_{x \in \{-1, 1\}^n} f(x) \phi(x) - \eps\sum_{x \in \{-1, 1\}^n} |\phi(x)| > w, \end{equation}
\begin{equation} \label{eq:owprelim1} \left|\sum_{x \in \{-1, 1\}^n} \phi(x) \chi_S(x)\right| \le 1   \text{ for each } 0 < |S| \leq d,\end{equation}
\begin{equation} \label{eq:owprelim2} \sum_{x \in \{-1, 1\}^n} \phi(x) = 0, \end{equation}
\begin{equation} \label{eq:owprelim3} \phi(x) \le 0 \text{ for each } x \in f^{-1}(-1). \end{equation}
\end{theorem}

\subsubsection{Threshold Degree and Threshold Weight}
We say a polynomial $p(x) = \sum_S c_S \chi_S(x)$ with \emph{integer} coefficients is a polynomial threshold function (PTF) for a Boolean
function $f$
if $p$ sign-represents $f$ at all Boolean inputs, i.e., if $f(x)p(x) > 0$ for all $x \in \{-1, 1\}^n$. The \emph{threshold degree} of $f$, $\tdeg(f)$, is the minimum degree of a PTF for $f$.

The \emph{threshold weight} $W(f)$ is the minimum weight of any PTF for $f$.  Observe that this definition is only meaningful because the coefficients of any PTF for $f$ are required to be integers, as any positive constant multiple of a PTF for $f$ also sign-represents $f$. 
More generally, it is of interest to study the tradeoff between the weight and degree necessary for PTF representations. To
this end, we define the \emph{degree-$d$ threshold weight} $W(f, d)$ to be the minimum weight of a degree-$d$ PTF for $f$. If $\tdeg(f) > d$, define $W(f, d) = \infty$.

While threshold weight is naturally captured by an \emph{integer} program rather than a linear program, it still admits an important dual characterization, obtained by combining results of Freund \cite{freund} and Hajnal et al. \cite{hajnal} (see also \cite{ghr, patternmatrix}).

\begin{theorem} \label{thm:twdual}
Let $f: \bits^n \to \bits$ and fix an integer $d \ge \tdeg(f)$. Then for every probability distribution $\mu$ on $\bits^n$,
\begin{equation} \label{eq:twprelim0} \left|\E_{x \sim \mu} [f(x)\chi_S(x)] \right| \ge \frac{1}{W(f, d)} \text{ for each } |S| \leq d.\end{equation}
Moreover, there exists a distribution $\mu$ for which
\begin{equation} \label{eq:twprelim1}
\left|\E_{x \sim \mu} [f(x)\chi_S(x)] \right| \le \left(\frac{2n}{W(f, d)}\right)^{1/2} \text{ for each } |S| \leq d.
\end{equation}
\end{theorem}

\subsection{Relating Degree-$d$ Threshold Weight to High-Error Approximations}
\label{sec:weightrelations}
In this paper, we will often need to translate
lower bounds on $\adeg_{\eps}(f)$ for some function $f$ with $\eps$
very close to 1
into lower bounds on the degree-$d$ threshold weight of $f$.
This is possible because degree-$d$ PTFs of weight $w$ are closely related to degree-$d$ pointwise approximations with error $1-1/w$. In fact, these notions are essentially equivalent when $w \ge {n \choose d}$ \cite{patternmatrix}. The relationships we will need are formalized
in the following lemma.

\begin{lemma}\label{lem:weight-relations}
Let $f: \bits^n \to \bits$ be a Boolean function, and let $w > 0$. Then (1) $\Rightarrow$ (2) $\Rightarrow$ (3).
\begin{enumerate}[(1)]
\item $\adeg_{1 - \frac{1}{w}}(f) > d$.
\item $W_{1 - \frac{1}{w}}(f, d) > 1$.
\item $W(f, d) > w$.
\end{enumerate}
\end{lemma}

\eat{\begin{theorem}[\cite{patternmatrix}, Theorem 2.5]
Let $f: \{-1, 1\}^n \to \{-1, 1\}$ be a Boolean function. Denote by $E(f, d)$ the minimum value of $\|p -f\|_\infty$ over real polynomials $p$ of degree $d$. Then for each $d$,
\[\frac{1}{1-E(f, d)} \le W(f, d) \le \frac{2}{1 - E(f, d)}\left({n \choose 0} + \dots + {n \choose d}\right)^{3/2}.\]
\end{theorem}}

Lemma \ref{lem:weight-relations} implies that a PTF of degree $d$ and weight $w$ can be transformed into $(1-1/w)$-approximation of degree $d$. Indeed, the proof will go by way of such a transformation.

\begin{proof}
Clearly (1) implies (2), since $W_{1 - \frac{1}{w}}(f, d) = \infty$ when $\adeg_{1 - \frac{1}{w}}(f) > d$. To show that (2) implies (3), 
 suppose  there is a PTF $p$ for $f$ having weight $w$ and degree $d$. Since $p$ has integer coefficients and is nonzero on Boolean inputs, $|p(x)| \ge 1$ on $\bits^n$. Moreover, $|p(x)| \le w$ by the weight bound, so the polynomial $\frac{1}{w}p(x)$ is a $(1 - \frac{1}{w})$-approximation to $f$ with weight $1$.
\end{proof}

\medskip
\noindent \textbf{Remark:} We stress that the converse of 
Lemma \ref{lem:weight-relations} fails badly
when $w \ll {n \choose d}$. For example,
we show in Corollary \ref{cor:colt}
that for any $d > 0$ there exists a read-once DNF
$F$ satisfying $W(F, d) \geq \exp\left(\sqrt{n/d}\right)$.
In particular, this yields an exponential lower bound
on the degree-$d$ threshold weight of $F$ for any $d=n^{1-\delta}$, with $\delta > 0$ a constant.
Yet it follows from a result of Sherstov \cite{sherstovrobust}
that $\adeg_{1/3}(F) = O(n^{1/2})$ for any read-once DNF $F$.

\eat{Conversely, for $W \ge {n \choose d}$, a $(1-1/W)$-approximation of degree $d$ can be turned into a PTF of degree $d$ and weight $\poly(W)$.

This theorem essentially optimal, in the sense that when $W \ll {n \choose d}$, all degree-$d$ PTFs for $f$ may have weight much larger than $W$,
even though there is a degree-$d$ polynomial that approximates $f$ to high accuracy. For example, $\widetilde{\deg}(\omb) = \tilde{O}(\sqrt{n})$ i.e.
there is a degree $\tilde{O}(\sqrt{n})$ polynomial that approximates \omb\ to constant accuracy, but 
a lower bound of Servedio, Tan, and Thaler \cite{thalercolt} shows that any degree $d=\sqrt{n}$ PTF for $\omb$ requires weight $2^{\Omega(\sqrt{n/d})} = 2^{\Omega(n^{1/4})}$.}

\section{Hardness Amplification for Approximate Degree}
\label{sec:highlow}
In this section, we show how to generically transform a circuit $f$ with one-sided approximate degree $d$ into
a circuit $F$ with $\eps$-approximate degree $d$ for $\eps=1-2^{-t}$. That is, while $f$ cannot be approximated to error 1/2 by
degree $d$ polynomials, $F$ cannot even be approximated to error $1-2^{-t}$ by polynomials of the same degree.

\begin{customthm}{\ref{thm:big}}
Let $f: \{-1, 1\}^m \rightarrow \{-1, 1\}$ be a function with $\odeg_{1/2}(f) > d$.
Let $F: \{-1, 1\}^{mt} \to \{-1, 1\}$ denote the function $\OR_{t}(f, \dots, f)$. Then $\adeg_{1-2^{-t}}(F) > d$.
\end{customthm}

We remark that it is necessary that the \emph{one-sided} approximated degree of $f$ is large,
rather than that just the approximate degree of $f$ is large.
Theorem \ref{thm:big} is easily seen to be false with one-sided approximate degree replaced by approximate degree.
Consider for example the case where $f=\OR_m$. Then $F=\OR_t(\OR_{m}, \dots, \OR_m) = \OR_{mt}$.
Since $\adeg(\OR_m) = \Omega(\sqrt{m})$, 
applying Theorem \ref{thm:big} with $\adeg$ in place of $\odeg$ would say that $\adeg_{1-2^{-t}}(\OR_{mt}) = \Omega(\sqrt{mt})$.
Yet the polynomial $q(y) = \frac{1}{mt}(1/2 - \sum_{i = 1}^t \sum_{j=1}^m y_{ij})$ demonstrates that $\adeg_{1-\frac{1}{2mt}}(\OR_{mt}) = 1$
for all values of $t$.
However, Theorem \ref{thm:big} does not apply because the one-sided approximate degree of $f=\OR_m$ is constant.

\begin{proof} 
Let $\psi$ be a dual polynomial for $f$ with one-sided error whose existence is guaranteed by the assumption that $\odeg_{1/2}(f) > d$. 
By \thmref{thm:oprelim}, $\psi$ satisfies:

\begin{equation} \label{eq:oproof0} \sum_{x \in \{-1, 1\}^m} \psi(x)f(x) > 1/2, \end{equation}
\begin{equation} \label{eq:oproof1} \sum_{x \in \{-1, 1\}^m} |\psi(x)| = 1,\end{equation} 
\begin{equation} \label{eq:oproof2} \sum_{x \in \{-1, 1\}^m} \psi(x) \chi_S(x)=0   \text{ for each } |S| \leq d \text{ and } \end{equation}
\begin{equation} \label{eq:oproof3} \psi(x) \le 0   \text{ for each } x \in f^{-1}(-1).\end{equation}

We will construct a dual solution $\zeta$ that witnesses the fact that $\adeg_{1-2^{-t}}(F) > d$. Specifically, $\zeta$ must satisfy the
three conditions of \thmref{thm:prelim}:

\begin{equation} \label{eq:show1} \sum_{(x_1, \dots, x_{t}) \in \left(\{-1, 1\}^{m}\right)^t} \zeta(x_1, \dots, x_{t}) F(x_1, \dots, x_{t}) > 1-2^{-t}.\end{equation}
\begin{equation} \label{eq:show2} \sum_{(x_1, \dots, x_{t}) \in \left(\{-1, 1\}^{m}\right)^t}  |\zeta(x_1, \dots, x_{t})| = 1. \end{equation}
\begin{equation} \label{eq:show3}  \sum_{(x_1, \dots, x_{t}) \in \left(\{-1, 1\}^{m}\right)^t}  \zeta(x_1, \dots, x_{t})\chi_S(x_1, \dots, x_t) =0   \text{ for each } |S| \leq d.\end{equation}

The construction of $\zeta$ is as follows. Let $\mathbf{1}$ denote the all-ones vector.
Let $\Psi: \{-1, 1\}^t \rightarrow \{-1, 1\}$ be defined such that $\Psi(\mathbf{1})=1/2$, $\Psi(-\mathbf{1})=-1/2$, and $\Psi(x)=0$ for all other $x$. 
Notice that 

\begin{equation}  \sum_{(x_1, \dots, x_{t}) \in \{-1, 1\}^t} \Psi(x_1, \dots, x_{t}) = 0 \label{eq:balanced} \end{equation}

We define $\zeta: \left(\{-1, 1\}^{m}\right)^{t} \rightarrow \mathbb{R}$ by
\begin{equation} \label{eq:zeta} \zeta(x_1, \dots, x_{t}) := 2^t \Psi(\dots, \sgn(\psi(x_i)), \dots) \prod_{i=1}^{t} |\psi(x_i)|,\end{equation}
where $x_i = (x_{i, 1}, \dots, x_{i, m})$.

\eqref{eq:zeta} combines dual functions $\Psi$ and $\psi$ to obtain a dual witness $\zeta$ in exactly the same manner
as in the works of Sherstov \cite[Theorem 3.3]{sherstovFOCS} and Lee \cite{lee}.
The analysis 
in these works implies without modification that $\zeta$ satisfies Equations \eqref{eq:show2} and \eqref{eq:show3}. That is, these works show

\begin{claim} \label{claim:show2}
\[\sum_{(x_1, \dots, x_{t}) \in \left(\{-1, 1\}^{m}\right)^t}  |\zeta(x_1, \dots, x_{t})| = 1.\]
\end{claim}

\begin{claim} \label{claim:show3}
\[\sum_{(x_1, \dots, x_{t}) \in \left(\{-1, 1\}^{m}\right)^t}  \zeta(x_1, \dots, x_{t})\chi_S(x_1, \dots, x_t) =0   \text{ for each } |S| \leq d.\]
\end{claim}

We provide
this analysis in Appendix \ref{app:proof1} for completeness, and here focus on arguing that (\ref{eq:show1}) holds. 
As we remarked earlier, the properties we exploit to show this are 
(1) that $\psi$ has one-sided error and (2) that the the vector
$-\mathbf{1}$ has Hamming distance $t$
from the (unique) input in $\OR_t^{-1}(1)$.

We now prove that (\ref{eq:show1}) holds. Let $\mu$ be the distribution on $\left(\{-1, 1\}^{m}\right)^t$ given by $\mu(x_1, \dots, x_t) = \prod_{i=1}^{t} |\psi(x_i)|$. 
Since $\psi$ is orthogonal to the constant polynomial, it has expected value 0, and hence the string $(\dots, \sgn(\psi(x_i)), \dots)$ is distributed uniformly in $\{-1, 1\}^{t}$
when one samples $(x_1, \dots, x_{t})$ according to $\mu$. 
 Observe that
$$ \sum_{(x_1, \dots, x_{t}) \in \left(\{-1, 1\}^{m}\right)^t} \zeta(x_1, \dots, x_{t}) F(x_1, \dots, x_{t})$$
$$= 2^{t} \mathbf{E}_\mu [\Psi( \dots, \sgn(\psi(x_i)), \dots) \OR_t\left( \dots, f(x_i), \dots\right)]$$
\begin{equation} \label{eq1andor} = \sum_{z \in \{-1, 1\}^{t}} \Psi(z) \left(\sum_{(x_1, \dots, x_{t}) \in \left(\{-1, 1\}^m\right)^t} \OR_t\left( \dots,  f(x_i), \dots\right) \mu(x_1, \dots, x_{t}|z)\right),\end{equation}
where $\mu(\mathbf{x}|z)$ denotes the probability of $\mathbf{x}$ under $\mu$, conditioned on $(\dots, \sgn(\psi(x_i)), \dots)=z$.

Let $A_{1} = \{x \in \{-1, 1\}^{m}: \psi(x) > 0, f(x) = -1\}$ and $A_{-1} = \{x \in \{-1, 1\}^{m}: \psi(x) < 0, f(x) = 1\}$, so
$A_1 \cup A_{-1}$ is the set of all inputs $x$ where the sign of $\psi(x)$ disagrees with $f(x)$. Notice that 
$\sum_{x \in A_1 \cup A_{-1}} |\psi(x)| < 1/4$ because $\psi$ has correlation $1/2$ with $f$. 

Let $\lambda$ be the distribution on $\{-1, 1\}^m$ defined by $\lambda(x) = |\psi(x)|$. Then for any bit $b$,
\[\Pr_{x \sim \lambda} [f(x) \ne \sgn(\psi(x)) | \sgn(\psi(x)) = b] = 2\sum_{x \in A_b} |\psi(x)|.\]
Therefore, as noted in \cite{sherstovFOCS}, for any given $z \in \{-1, 1\}^{t}$, the following two random variables are identically distributed:

\begin{itemize}
\item The string $(\dots, f(x_i), \dots)$ when one chooses $(\dots, x_i, \dots)$ from the conditional distribution 
$\mu(\cdot|z)$. 
\item The string $(\dots, y_iz_i, \dots)$, where $y \in \{-1, 1\}^{t}$ is a random string whose $i$th bit independently
takes on value $-1$ with probability $2 \sum_{x \in A_{z_i}} |\psi(x)| < 1/2$. 
\end{itemize}

Thus, Expression (\ref{eq1andor}) equals

\begin{equation} \label{eq2andor} \sum_{z \in \{-1, 1\}^{t}} \Psi(z) \cdot \mathbf{E}[\OR_t(\dots, y_iz_i, \dots)],\end{equation}

where $y \in \{-1, 1\}^{t}$ is a random string whose $i$th bit independently
takes on value $-1$ with probability $2 \sum_{x \in A_{z_i}} |\psi(x)| < 1/2$.
We first argue that the term corresponding to $z=\mathbf{1}$ contributes $\Psi(z) = 1/2$ to Expression (\ref{eq2andor}).
By \eqref{eq:oproof3}, if $f(x) = -1$, then $\psi(x) \le 0$. This implies that $A_{1}$ is empty; that is, if $\sgn(\psi(x))=1$,
then it must be the case that $f(x)=1$. Therefore, for $z=\mathbf{1}$, the $y_i$'s are all $1$ with probability 1, and hence $\mathbf{E}_y[\OR_t\left(\dots, y_iz_i, \dots\right)] = \OR_t\left(\mathbf{1}\right) = 1$.
Thus the term corresponding to $z=\mathbf{1}$ contributes $\Psi(z)\OR_t(z)=1/2$ to Expression (\ref{eq2andor}) as claimed.

All $z \not\in \{\mathbf{1}, \mathbf{-1}\}$ are given zero weight by $\Psi$ and hence contribute nothing to the sum. 
All that remains is to show that the contribution of the term $z=-\mathbf{1}$ to the sum is $\frac{1}{2} (1- 2^{-t})$. 
Since each $y_i=1$ independently with probability at least $1/2$, and $\OR_t(\dots, -y_i, \dots)=-1$ as long as there is at least one $y_i \neq -1$,
we conclude that  $\mathbf{E}[\OR_t(\dots, y_iz_i, \dots)] \geq 1-2^{-t+1}$. It follows that  the term corresponding to
$z=-\mathbf{1}$ contributes at least $\frac{1}{2}(1 - 2^{-t+1})$ to the sum. Thus,
$$\sum_{z \in \{-1, 1\}^{t}} \Psi(z) \cdot \mathbf{E}[\OR_t(\dots, y_iz_i, \dots)] \geq \frac{1}{2} + \frac{1}{2} (1-2^{-t+1}) = 1-2^{-t}.$$ 
This completes the proof.
\end{proof}

\noindent \textbf{Remark:} Since the set $A_1$ within the proof of Theorem \ref{thm:big} is empty, the ``combined'' dual witness $\zeta$ constructed in the proof
in fact has one-sided error. Thus, the proof establishes
that $\odeg_{1-2^{-t}}(F) > d$, which is a stronger conclusion
than the $\adeg_{1-2^{-t}}(F) > d$ bound appearing in the theorem statement. 
We chose to state Theorem \ref{thm:big} as an approximate degree
lower bound, rather than as a one-sided approximate degree lower bound,
for easier comparison with prior work on approximate degree.

\section{Lower Bounds for \acz} \label{sec:aczlb}

In this section, we establish a new lower bound on the one-sided approximate degree of \acz. 
Combining this lower bound with Theorem \ref{thm:big}, we establish new lower bounds on accuracy vs. degree tradeoffs for \acz.
This in turn yields a new upper bound on the discrepancy, and a new lower bound on the threshold weight of \acz.

\subsection{The One-Sided Approximate Degree of Symmetric Properties}
\label{sec:symmetric}

We identify a fairly general criterion under which the one-sided approximate degree of a Boolean function is equal to its approximate degree. This criterion applies to many functions previously studied in the literature, including the $\AND$ function, the \ED\ and {\sc Collision}  functions \cite{aaronsonshi, ambainis}, and the \tto\ function \cite{beame}. Our result applies to Boolean functions corresponding to \emph{symmetric properties}; we refer the reader to Section \ref{sec:intro-odeg-acz} for the relevant notation and definitions.

\thmsymmetric*

\begin{proof}
Suppose $\odeg_\eps(f) = d$. Let $p$ be any degree-$d$ one-sided approximation to $f$ with error $\eps$.
As described in the proof overview in Section \ref{sec:intro-odeg-acz},
we show how to transform $p$ into 
an actual $\eps$-approximation $r$ for $f$ such that $\deg r \leq (\log_2 R)\deg p$.
Our transformation from $p$ to $r$ consists of two steps.

In the first step,
we turn $p$ into a ``symmetric'' polynomial $p^{\text{sym}}(x)$ defined below. The following symmetrization lemma shows that the map $p \mapsto p^{\text{sym}}$ increases the degree of $p$ by at most a factor of $\log_2 R$.

\begin{restatable}{lemma}{lemambainissym} \label{lem:ambainis-sym}
Let $m=N \cdot \log_2 R$. For $x, y \in \{-1, 1\}^m$, write $y \sim x$ if there is a pair of permutations $\sigma$ on $[R]$ and $\pi$ on $[N]$ such that $g_y = \sigma \circ g_x \circ \pi$. Let $p : \{-1, 1\}^m \to \R$ be a real polynomial. Define
\[p^{\mathrm{sym}}(x) = \E_{y \sim x} [p(y)].\]
Then $\deg(p^{\mathrm{sym}}) \le (\log_2 R) \deg(p)$.
\end{restatable}

The proof of Lemma \ref{lem:ambainis-sym} exploits a result of Ambainis \cite{ambainis} and appears in Appendix \ref{app:ambainis}.

\medskip
We now turn to the second step of our transformation,
in which we identify an affine transformation $r$ of $p^{\text{sym}}$ that is an actual $\eps$-approximation to $f$.
To this end, we make two further observations about the polynomial $p^{\mathrm{sym}}$.

\begin{claim}\label{claim:sym-approx}
If $\phi_f$ is a symmetric property and $p$ is an $\eps$-one-sided approximation to $f$, then $p^{\mathrm{sym}}$ is also an $\eps$-one-sided approximation to $f$.
\end{claim}

\begin{claim}\label{claim:sym-trans}
Let $S \subseteq \{-1, 1\}^m$. If $x \sim y$ for every pair $x, y \in S$, then $p^{\mathrm{sym}}$ is constant on $S$.
\end{claim}

We first show how these claims together imply the theorem. By Claim \ref{claim:sym-approx}, $p^{\mathrm{sym}}$ is an $\eps$-one-sided approximation to $f$. By Claim \ref{claim:sym-trans}, $p^{\mathrm{sym}}$ is constant on the set of inputs $f^{-1}(-1)$, where it takes some value $v \le -1 + \eps$. If $v \ge -1 - \eps$, then $p^{\mathrm{sym}}$ is itself an $\eps$-approximation to $f$ and we are done. Otherwise, define the polynomial
\[r(x) = 1 + \frac{2(p^{\mathrm{sym}}(x) - 1)}{|v - 1|}.\]
Then $r(x) = -1$ for all $x \in f^{-1}(-1)$. Moreover, since $|v - 1| \ge 2$, we have $r(x) \in [1 - \eps, 1 + \eps]$ for all $x \in f^{-1}(1)$. Thus $r$ is a $\eps$-approximation to $f$.

We now proceed to prove Claims \ref{claim:sym-approx} and \ref{claim:sym-trans}.

\begin{proof}[Proof of Claim \ref{claim:sym-approx}]
Suppose that $\phi_f$ is a symmetric property and that $p$ is an $\eps$-one-sided approximation to $f$. Then $p(x) \in [1 - \eps, 1 + \eps]$ for all $x \in f^{-1}(1)$, and $p(x) \le -1 + \eps$ for all $x \in f^{-1}(-1)$. We focus on some fixed $x \in f^{-1}(1)$; handling inputs in $f^{-1}(-1)$ is entirely analogous. Since $\phi_f$ is symmetric, we have $y \in f^{-1}(1)$ whenever $y \sim x$. Therefore, $p(y) \in [1-\eps, 1+\eps]$ whenever $y \sim x$, so
\[p^{\mathrm{sym}}(x) = \E_{y \sim x} [p(y)] \in [1 - \eps, 1 + \eps].\]
A similar argument holds for $x \in f^{-1}(-1)$, so $p^{\mathrm{sym}}$ is an $\eps$-one-sided approximation to $f$.
\end{proof}

\begin{proof}[Proof of Claim \ref{claim:sym-trans}]
Fix a set $S \subseteq \{-1, 1\}^m$, and suppose $x \sim y$ for every $x, y \in S$. Fix some $x^* \in S$, and let $x \in S$ be arbitrary. It suffices to show that $p^{\mathrm{sym}}(x) = p^{\mathrm{sym}}(x^*)$. Write $\sigma \cdot x \cdot \pi$ for the value $y$ for which $g_y = \sigma \circ g_x \circ \pi$. Let $\sigma_x, \pi_x$ be a pair of permutations where $x = \sigma_x \cdot x^* \cdot \pi_x$. Note that the map $\sigma \mapsto \sigma \circ \sigma_x$ is a bijection from the symmetric group over $[R]$ to itself; similarly the map $\pi \mapsto \pi_x \circ \pi$ is a bijection from the symmetric group over $[N]$ to itself. Hence, it holds that
\begin{align*}
\E_{y \sim x} [p(y)] &= \E_{\sigma, \pi} [p(\sigma \cdot x \cdot \pi)] \\
&= \E_{\sigma, \pi} [p((\sigma \circ \sigma_x) \cdot x^* \cdot (\pi_x \circ \pi))] \\
&= \E_{\sigma, \pi} [p(\sigma \cdot x^* \cdot \pi)] \\
&= \E_{y \sim x^*}[p(y)].
\end{align*}
Thus we have $p^{\mathrm{sym}}(x) = p^{\mathrm{sym}}(x^*)$, showing that $p^{\mathrm{sym}}$ is constant on $S$.
\end{proof}

The proof of these claims concludes the proof of Theorem \ref{thm:symmetric}.
\end{proof}

\eat{Define a polynomial $q':[N] \to \R$ by
\[q'(z) = \E_{\tau \in S^R}[q(z_{\tau(1)}, \dots, z_{\tau(R)})].\]
Then $q'$ has degree at most $d$, and since $\phi_f$ is a symmetric property, $q'(z(x))$ is still a one-sided approximation to $f$. }

\subsection{The One-Sided Approximate Degree of \acz}
\label{sec:odegacz}
Prior to our work, the best lower bound on the one-sided approximate degree of an \acz\ function on $m$ variables was $\Omega(\sqrt{m})$ attained by the $\AND_m$ function (Fact \ref{fact:odegand}). However, to obtain stronger lower bounds for \acz\ via our hardness amplification technique, we need a constant-depth circuit with polynomially larger one-sided approximate degree. We now exhibit a depth-two circuit having one-sided approximate degree $\tilde{\Omega}(m^{2/3})$. Let $N$ and $R$ be positive integers such that $N \geq R$ and $R$ is a power of 2. We define the \ED\ function with range $R$ as follows. The function takes $m = N \log_2 R$ bits as input,
and interprets its input as $N$ blocks $(x_1, \dots, x_N)$ with each
block consisting of $\log_2 R$ bits. Each block is interpreted as 
a number in the range $[R]$, and the function
evaluates to TRUE if and only if all $N$ numbers are distinct. 

It is straightforward to check that for $R=\poly(N)$, the \ED\ function
with range $R$ is computed by a CNF formula
of polynomial size. Indeed, the function evaluates to TRUE if and only if there is no number $K \in [R]$ for which there is a pair of distinct indices $i, j \in [N]$ such that $x_i = x_j = K$. Thus, the following natural CNF computes \ED\ (noting that for any fixed $K$, the inner formula is computed by a bitwise OR):
\[f(x_1, \dots, x_N) = \bigwedge_{K = 1}^R \bigwedge_{i \ne j} (x_i \ne K) \lor (x_j \ne K).\]

Aaronson and Shi \cite{aaronsonshi} showed that
when $R > 3N/2$, the approximate degree of \ED\ is $\Omega(N^{2/3})$. 
Ambainis \cite{ambainis} extended the lower bound to the ``small-range''
case where $R=N$. For the remainder of the paper, we will use the term
\ED\ without qualification to refer to the small-range case.

In the language of properties of functions, the property $\phi_{\text{ED}}$ determines whether a function is one-to-one. That is, $\phi_{\text{ED}}(g) = 1$ if and only if the function $g: [N] \to [R]$ is injective. Note that $\phi_{\text{ED}}$ is a symmetric property, since injectivity is preserved under permutations of a function's domain and range. Furthermore, it is straightforward to verify
that \ED\ satisfies the hypothesis of Theorem \ref{thm:symmetric}: for any two inputs $x, y$ corresponding
to one-to-one functions $g_x, g_y : [N] \rightarrow [R]$, there exist permutations $\sigma, \pi$ such that 
$g_y = \sigma \circ g_x \circ \pi$.  
Thus, Theorem \ref{thm:symmetric} implies that the one-sided approximate degree of \ED\ is $\Omega(N^{2/3}/\log R)$. 

In a prior version of this work, we gave a different proof of this fact for the small-range case $N = R$ by manipulating a dual witness for the high approximate degree of \ED. We provide this alternative argument in Appendix \ref{app:as}. In fact, in the small-range case, the fact the property $\phi_{\text{ED}}$ holds for exactly one input function up to permutation \emph{of the domain only} allows us to prove a stronger one-sided approximate degree lower bound of $\Omega(N^{2/3})$ (i.e, without the loss of a $\log R$ factor that arises from symmetrizing over the range in the proof of Theorem \ref{thm:symmetric}).

\cored*

\subsection{Accuracy vs. Degree Tradeoffs for \acz}

We are now in a position to prove
our new lower bound on ``accuracy vs. degree'' tradeoffs for pointwise approximating \acz\ functions by polynomials.
\coredamp*
\begin{proof}
Let $t=n/d^{3/2}$, and $m=d^{3/2}$. Define $F = \OR_t(f, \dots, f)$ where $f: \{-1, 1\}^m \to \{-1, 1\}$ computes the \ED\ problem. The discussion in Section \ref{sec:odegacz} implies that $f$ is computed by a depth-2 circuit, and that $f$ has one-sided approximate degree $\tilde{\Omega}(m^{2/3})$. The claim now follows by \thmref{thm:big}.
\end{proof}


\subsubsection{On the Tightness of \thmref{thm:big} and \corref{cor:ed-amp}} \label{sec:tightness}

We now argue that the approximate degree lower bound proved in \thmref{thm:big} is essentially tight. In particular, we show that the function $F$ for which \corref{cor:ed-amp} yields a $(1 - \exp(-\tilde{\Omega}(n^{2/5}))$-error lower bound for approximating polynomials of degree $n^{2/5}$ actually admits a $(1 - \exp(-\tilde{O}(n^{2/5}))$-approximating polynomial of degree $\tilde{O}(n^{2/5})$.

Our nearly-matching upper bound makes use of a well-known paradigm for constructing low-weight PTFs (and hence, by Lemma \ref{lem:weight-relations}, low-accuracy pointwise approximations) for composed functions by way of \emph{rational approximations} (see e.g. \cite{ppclosed, hshs}). Suppose $f: \bits^m \to \bits$ is pointwise approximated by a rational function in the sense that for every $x \in \bits^m$,
\[\left|f(x) - \frac{p(x)}{q(x)}\right| < \frac{1}{t},\]
where $p, q$ are polynomials of degree $d$ and weight $w$ and $q(x) > 0$ on $\bits^m$. Then observe that the block composition
\[\OR_t(f(x_1), \dots, f(x_t)) = \operatorname{sgn}(1 - t + f(x_1) + \dots + f(x_t)) = \operatorname{sgn}\left(1 - t + \frac{p(x_1)}{q(x_1)} + \dots \frac{p(x_t)}{q(x_t)}\right).\]
Multiplying $\left(1 - t + \frac{p(x_1)}{q(x_1)} + \dots \frac{p(x_t)}{q(x_t)}\right)$ by the positive quantity $q(x_1)\cdot \dots \cdot q(x_t)$ and clearing denominators yields a PTF for the composed function of degree $td$ and weight at most $w^t(m + tw)$.

We now construct a rational approximation for $f = $ \ED\ with the desired properties. Recall from \secref{sec:odegacz} that \ED\ on $m$ variables has a CNF representation where the top $\AND$ gate has fan-in $s := O(m^3)$ and each $\OR$ gate has fan-in $O(\log m)$. It is easy to check that $\AND_s: \{-1, 1\}^s \rightarrow \{-1, 1\}$ admits the rational approximation
\[\frac{ts - 1 + t\sum_{i=1}^s x_i}{ts + 1 + t\sum_{i=1}^s x_i}\]
with error $1/t$, degree $d = 1$, and weight $w = O(st)$. Moreover, each bottom $\OR$ gate in the CNF can be computed exactly by a degree $O(\log m)$ polynomial with weight $O(1)$. Composing these constructions yields a rational approximation for \ED\ with error $1/t$, degree $d = O(\log m) = O(\log t)$ and weight $O(st) = \poly(t)$. Therefore, $F$ has a PTF of degree $\tilde{O}(t)$ and weight $\exp(\tilde{O}(t))$. By the construction of \lemref{lem:weight-relations}, $F$ also has a $(1 - \exp(-\tilde{O}(t)))$-approximation of degree $\tilde{O}(t)$. When $t = n^{2/5}$, we obtain a $(1 - \exp(-\tilde{O}(n^{2/5}))$-approximating polynomial of degree $\tilde{O}(n^{2/5})$ for $F$ as claimed.

\subsubsection{A Sharp Threshold in Accuracy-Degree Tradeoffs} \label{sec:sharp-threshold}

The rational approximations developed in the previous section, combined
with the lower bound of Theorem \ref{thm:big} and Corollary \ref{cor:ed-amp}, reveal a ``sharp threshold'' in the degree required to approximate a particular function $F$ within a given error parameter.
Recall
that Theorem \ref{thm:big} and Corollary \ref{cor:ed-amp} yield a lower bound of $d=\Omega(m^{2/3}/\log m)$ on the $\eps$-approximate degree of  $F = \OR_t(f, \dots, f)$, where $f$ is the \ED\ function on $m$ variables and $\eps=1-2^{-t}$. In the following discussion, consider any $t = d^{1-\Omega(1)}$.

If our goal is to approximate $F$ to within error $(1 - \exp(-\tilde{O}(t)))$, then the rational approximation techniques described in the preceding section yield an approximating polynomial of degree $\tilde{O}(t)$. On the other hand, if we desire even slightly better error of $1 - 2^{-t}$, then our accuracy-degree tradeoff lower bound of \thmref{thm:big} shows that we require degree $d = \omega(t)$. That is, if we demand error that is slightly better than $1-\exp(-\tilde{O}(t))$, there is an asymptotic jump from $\tilde{O}(t)$ to $\Omega(d)$ in the required degree.

\eat{That is, there exists a degree $d = \Theta(n^{1/3})$ for which $F$ admits a PTF with weight only $2^{\tilde{O}(n^{1/3})}$, but if we decrease $d$ whatsoever, then $F$ has no PTF representation with such degree, regardless of weight.

The function with this property is the DNF $F: \bits^n \to \bits$ where $n = 4t^3$ is defined blockwise by
\[F(x_1, \dots, x_t) = \OR_t(\AND_{4t^2}, \dots, \AND_{4t^2}).\]
The discussion in the previous section shows that $F$ has a PTF of degree $t = n^{1/3}$ and weight $2^{\tilde{O}(t)}= 2^{\tilde{O}(n^{1/3})}$. On the other hand, Minsky and Papert \cite{mp} showed that $F$ has threshold degree at least $t = n^{1/3}$; that is, $F$ does not admit any PTF representation of degree smaller than $n^{1/3}$, even if we allow arbitrarily large weight.

This observation ties in to several important open questions (see the conclusion for further discussion). First, it would be interesting to resolve the question of whether low-degree and low-weight PTFs exist for general DNFs. As noted in \secref{sec:learning}, an affirmative answer would yield important consequences for learning via the Generalized Winnow algorithm. Second, it would be interesting to see if this sharp thresholding behavior holds more generally for \acz\ functions, particularly if we take $f = $ \ED. This relates to our conjecture (see the conclusion) that the threshold \emph{degree} of the circuit of \corref{cor:ed-amp} is $\tilde{\Omega}(n^{2/5})$.
}

\subsection{Discrepancy of \acz}
\label{sec:discrepancy}
In this section we prove our new exponentially small upper bound on the discrepancy of a function in \acz. Consider a Boolean function $f: X \times Y \rightarrow \{-1,1\}$, and let $M^{(f)}$ be its communication matrix $M^{(f)} = [f(x, y)]_{x \in X, y \in Y}$.
A combinatorial rectangle of $X \times Y$ is a set of the form $A \times B$ with $A \subseteq X$ and $B \subseteq Y$. 
For a distribution $\mu$ over $X \times Y$, the discrepancy of $f$ with respect to $\mu$ is defined to be the maximum over all rectangles $R$ of the \emph{bias} of $f$ on $R$. That is:
$$\disc_\mu(f) = \max_{R} \left|\sum_{(x, y) \in R} \mu(x, y) f(x, y)\right|.$$ The discrepancy of $f$, $\disc(f)$ is defined to be $\min_\mu \disc_\mu(f)$.

\eat{Discrepancy upper bounds directly yield lower bounds in a number of complexity models.
For instance, a discrepancy upper bound on a function $f: X \times Y \rightarrow \{-1, 1\}$ yield lower bounds in a variety of contexts,
including:

\begin{enumerate}
\item PP communication complexity, a model capturing small-bias randomized communication complexity.
\item Size lower bounds for majority-of-threshold circuits.
\item Base-set independent PTF weight lower bounds for the class of functions $C=\{f_x(y): Y \rightarrow \{-1, 1\}\}$, where $f_x(y):=f(x, y)$.
\end{enumerate}

In \cite{majmaj}, Sherstov exhibited the first exponentially small upper bound on the discrepancy of a function in \acz. As a consequence, he showed that majority-of-threshold circuits for \acz\ require exponential size, resolving a question of Krause and Pudl\'{a}k. Prior to this work, the best-known upper bound on the discrepancy of a function in \acz\ was $\exp\left(-\Omega(n^{1/3})\right)$, due to Buhrman et al. \cite{bvdw} and Sherstov \cite{patternmatrix}.
}

Sherstov's pattern matrix method \cite{patternmatrix} shows how to generically transform an \acz\ function with high threshold degree or high threshold weight into another \acz\ function with low discrepancy.

\begin{theorem}[\cite{patternmatrix}, adapted from Corollary 1.2 and Theorem 7.3] \label{thm:patternmatrix}
Let $F: \bits^n \to \bits$ be given, and define the communication problem $F': \bits^{4n} \times \bits^{4n} \to \bits$ by
$$F'(x, y) = F(\dots, \lor_{j=1}^4(x_{i, j} \land y_{i, j}), \dots ).$$
Then for every integer $d \ge 0$,
\[\disc(F')^2 \le \max \left\{ \frac{2n}{W(F, d-1)}, 2^{-d}\right\}.\]
\end{theorem}

We apply this theorem to the function $F: \bits^n \to \bits$ of \corref{cor:ed-amp}. This function has $\eps$-approximate degree $n^{2/5}$ for $\eps=1-2^{-\tilde{\Omega}(n^{2/5})}$, and hence by  by \lemref{lem:weight-relations} it holds that $W(f, n^{2/5}) = 2^{\tilde{\Omega}(n^{2/5})}$. We thus obtain our new discrepancy upper bound for \acz\ as stated in Corollary \ref{cor:discrepancy}, restated here for the 
reader's convenience. 

\cordisc*

\subsection{Threshold Weight of \acz}
\label{sec:ptfweight}
Combing Lemma \ref{lem:weight-relations} with 
Corollary \ref{cor:ed-amp} yields Corollary \ref{cor:ed-ampthresh},
restated here for the reader's convenience.

\coredampthresh*

A result of Krause \cite{krause} allows us to extend our new degree-$d$ threshold weight lower bound for \acz\ into an $\exp\left(\tilde{\Omega}\left(n^{2/5}\right)\right)$ \emph{degree independent} threshold weight lower bound for a related function $F'$. In \lemref{lem:threshold-weight} below, we give a slight modification of Krause's original result that is cleaner to apply, and asymptotically recovers Krause's result when the weights under consideration are superpolynomially large. Our restatement admits a new and simple proof based on LP duality  that we present in \appref{app:threshold-weight}. 

\eat{
\begin{lemma}[\cite{krause}, Lemma 3.4] \label{lem:threshold-weight}
Let $F: \bits^n \to \bits$ be a Boolean function, and define $F': \bits^{3n} \to \bits$ by
\[F'(x_1, \dots, x_n, y_1, \dots, y_n, z_1, \dots, z_n) := F(\dots, (\bar{z}_i \land x_i) \lor (z_i \land y_i), \dots).\]
Then $W(F') \ge W(F, d)$ for all $d$ for which $2^d \ge W(F, d)$.
\end{lemma}
}

\begin{restatable}{lemma}{lemthresholdweight} \label{lem:threshold-weight}
Let $F: \bits^n \to \bits$ be a Boolean function, and define $F': \bits^{3n} \to \bits$ by
\[F'(x_1, \dots, x_n, y_1, \dots, y_n, z_1, \dots, z_n) := F(\dots, (\bar{z}_i \land x_i) \lor (z_i \land y_i), \dots).\]
Then for every integer $d \ge 0$,
\[W(F')^2 \ge \min \left\{ \frac{W(F, d)}{2n}, 2^d\right\}.\]
\end{restatable}

Combining Corollary \ref{cor:ed-ampthresh} and Lemma \ref{lem:threshold-weight} yields Corollary \ref{cor:weightacz}.
This improves over the previous
best threshold weight lower bound for \acz, which was $\exp\left(\Omega(n^{1/3})\right)$ \cite{krausepudlak}.

\corweightacz*

\begin{proof}
 Let $F$ be the circuit of \corref{cor:ed-amp} and let $F'$ be the depth-four circuit obtained by applying Lemma \ref{lem:threshold-weight} to $F$. Let $d = n^{2/5}/\log^c n$ for a sufficiently large constant $c$. Then Corollary \ref{cor:ed-ampthresh} implies that $W(F, d) \ge 2n2^d$,
 and hence $W(F') \geq 2^{d/2} = 2^{\tilde{\Omega}(n^{2/5})}$
 by Lemma \ref{lem:threshold-weight}.
\end{proof}

\medskip \noindent \textbf{Remark:} While the threshold weight bound 
of Corollary \ref{cor:weightacz} is stated for polynomial
threshold functions over $\{-1, 1\}^n$ (i.e., for polynomials that are integer linear combinations of parities), the same threshold weight
lower bound also holds for polynomials over $\{0, 1\}^n$, or equivalently, for integer linear combinations of conjunctions.
This can be seen as follows.

Given a set $S \subseteq [n]$,
let  $\AND_S: \{-1, 1\}^n \rightarrow \{-1, 1\}$ denote the AND function restricted
to variables in $S$. Given a sign-representation $p = \sum_{S} c_S\AND_S$ for $F$ of weight $w$, let $\sum_{S} \hat{p}(S) \chi_S$
 denote the Fourier representation of $p$. 
It is easy to check that the $L_1$-norm of the Fourier coefficients of each conjunction $\AND_S$ is at most $3$, so the weight of the Fourier expansion of $p$ is $w' := \sum_{S} |\hat{p}(S)| \leq 3w$. However, we cannot simply conclude that $w/3 \ge w' \ge W(f)$ because
the coefficients $\hat{p}(S)$ are not necessarily integers. 

Nonetheless, note that $|p(x)| \geq 1$ for all $x \in \bits^n$,
since $p$ has integer coefficients. That is, $p$ is a sign-representation
for $f$ over $\{-1, 1\}^n$ of weight $w'$ and with margin
at least 1. It follows by Theorem \ref{thm:twdual}
that $\exp\left(\tilde{\Omega}(n^{2/5})\right) = W(f) \leq 2n(w')^2 = \poly(n, w)$. We conclude that $w=\exp\left(\tilde{\Omega}(n^{2/5})\right)$
as desired.

The same argument shows that all
of our lower bounds on degree-$d$ threshold weight proved in
this paper hold for PTFs over $\{0, 1\}^n$, in addition
to PTFs over $\{-1, 1\}^n$.

\section{Lower Bounds for AND-OR Trees}
\label{sec:andortree}
The $d$-level AND-OR tree (respectively, OR-AND tree) on $n$ variables 
is a function described by a read-once circuit
of depth $d$ consisting of alternating layers of AND gates and OR gates, with
the root gate being an AND gate (respectively, an OR gate).
We assume throughout this section that all gates have fan-in $n^{1/d}$;
for example, the two-level
AND-OR tree is a read-once CNF in which all gates have fan-in $n^{1/2}$.
The assumption on the fan-in is not essential to our analysis in this section,
which in fact applies to any read-once Boolean circuit such that all gates at any given layer have the same fan-in. 
We will let $\ANDOR_{d, n}$ (respectively, $\ORAND_{d, n}$) 
denote the $d$-level AND-OR tree (respectively, OR-AND tree) on $n$ variables.

The current authors \cite{bunthaler}, and independently Sherstov \cite{sherstovnew},
resolved the approximate degree of $\ANDOR_{2, n}$ by proving an optimal $\Omega(n^{1/2})$ lower bound in this case. 
However, the techniques of \cite{bunthaler, sherstovnew}
break down for the case of depth three or greater;
 to the best of our knowledge, the best lower bound
that follows from prior work is
$\Omega(n^{1/4+1/2d})$, which can be derived by combining the depth-two lower bound \cite{sherstovnew, bunthaler} with an earlier direct-sum theorem of Sherstov \cite[Theorem 3.1]{sherstovFOCS}. 

In this section, we extend the methods of our prior work \cite{bunthaler} to prove
an $\Omega\left(n^{1/2}/\log^{(d-2)/2} n\right)$ lower bound on the approximate degree of $\ANDOR_{d, n}$ for any constant $d>0$.

Up to a $\log^{(d-2)/2}n$ factor, this matches an upper bound of $O(n^{1/2})$ which was
established for AND-OR trees of any depth via a line of
work on quantum query algorithms \cite{hoyer, spalekandor, reichardt}.
Specifically, H\o yer, Mosca, and De Wolf \cite{hoyer} proved an upper bound
of $O(c^{d-1} n^{1/2})$ for some constant $c$
on the approximate degree of any depth-$d$ AND-OR tree
in which all gates at any given layer have the same fan-in. 
Subsequent work by Ambainis et al. \cite{spalekandor} established an upper bound of 
$n^{1/2}2^{O\left(\sqrt{\log n}\right)}$ for any depth, and further refinements
by Reichardt culminated in an $O(n^{1/2})$ upper bound for any depth \cite{reichardt}.
A remarkable result of Sherstov \cite{sherstovrobust}, on making polynomials robust to noisy inputs, yields a very different proof of
H\o yer, Mosca, and De Wolf's $O(c^{d-1}n^{1/2})$ upper bound. 

\thmandor*

 \begin{proof}
 We begin by proving the claimed lower bound for $\ANDOR_{3, n}$ before explaining how to extend the argument to 
 $\ANDOR_{d, n}$ for an arbitrary depth $d>0$. 

\paragraph{Notation.} There will be a total of seven intermediate dual witnesses that arise in our construction
 of a dual witness $\psi_7$ for $\ANDOR_{3, n}$. We will denote these
 seven dual witnesses as $\psi_1, \dots, \psi_7$.
  Let $M=n^{1/3}$ denote the fan-in of all gates in $\ANDOR_{3, n}$. Our goal is to construct a dual witness $\psi_7$
 to demonstrate that $\adeg(\ANDOR_{3, n})=\Omega\left(n^{1/2}/\log^{1/2}n\right)$.
 
To this end, let $\psi_6$ be a dual polynomial witnessing the fact that $\odeg_{.99}(\AND_{M})=\Omega(\sqrt{M})$. By \thmref{thm:prelim}, there is some $d_6 = \Omega(\sqrt{M})$ such that $\psi_6$ satisfies:

\begin{equation} \label{eq:andorproof0} \sum_{a \in \{-1, 1\}^M} \psi_6(a)\AND_{M}(a) > .99, \end{equation}
\begin{equation} \label{eq:andorproof1} \sum_{a \in \{-1, 1\}^M} |\psi_6(a)| = 1,\end{equation} 
\begin{equation} \label{eq:andorproof2} \sum_{a \in \{-1, 1\}^M} \psi_6(a) \chi_S(a)=0   \text{ for each } |S| \leq d_6 \text{ and} \end{equation}
\begin{equation} \label{eq:andorproof100} \psi_6(-\mathbf{1}) \le 0.\end{equation}

As stated in the proof outline (Section \ref{sec:introANDOR}), we are ultimately going to construct a function $\psi_5:\{-1, 1\}^{M^2} \rightarrow \mathbb{R}$ that serves as a dual witness to the high approximate
degree of $\ORAND_{2, M^2}$ while having ``almost no error on the wrong side''. More formally, we will show

\begin{claim} \label{claim:psi5}
There exists a dual witness $\psi_5$ for the fact that the $0.98$-approximate degree of $\ORAND_{2, M^2}$ is at least $\Omega(M/\sqrt{\log n})$ with the following property. If $A_{-1} = \{z \in \{-1, 1\}^{M^2} : \psi_5(z) < 0 \text{ and } \ORAND_{2, M^2}(z) = 1\}$, then
\[\sum_{z \in A_{-1}} |\psi_5(z)| \le n^{-2}.\]
\end{claim}

We will then define our final dual witness $\psi_7$ via
\begin{equation} \label{eq:zetaprimerepeat} \psi_7(x_1, \dots, x_{M}) := 2^{M} \psi_6(\dots, \sgn(\psi_5(x_i)), \dots) \prod_{i=1}^{M} |\psi_5(x_i)|,\end{equation}
where $x_i = (x_{i, 1}, \dots, x_{i, M^2})$. 

We now prove the existence of $\psi_5$ (Claim \ref{claim:psi5}) before returning to the analysis of the combined dual witness $\psi_7$.

\begin{proof}[Proof of Claim \ref{claim:psi5}]
As discussed in the introduction, the construction of $\psi_5$ combines our hardness amplification technique (Theorem \ref{thm:big}) with the technique of combining dual witnesses in our earlier work \cite{bunthaler}.

\medskip
\noindent \textbf{Construction of $\psi_5$}. 
Consider the function $\ORAND_{2, M^2}$. Let $t=2\log n$. We view the root OR gate as an OR of ORs, where the top OR has fan-in $M/t$ and the bottom OR gates each have fan-in $t$. Thus, we are now thinking of the two-level OR-AND tree as a three-level circuit, where the top two levels consist of OR gates, and the bottom level consists of AND gates. 
Consider the function $F=\OR_t(\AND_{M}, \dots, \AND_{M})$, which allows us to write $\ORAND_{2, M^2} = \OR_{M/t}(F, \dots, F)$. By our hardness amplification technique, there is a dual witness $\psi_3$ for the high one-sided approximate degree of $F$, even with error inverse polynomially close to $1$. We will construct $\psi_5$ by combining $\psi_3$ with a dual witness $\psi_4$ for the high approximate degree of $\OR_{M/t}$.

In more detail, applying Theorem \ref{thm:big} to the $\AND_M$ function (see Fact \ref{fact:odegand}) yields a dual witness $\psi_3$ demonstrating that there is some $d_3 = \Omega(\sqrt{M})$ such that
$\odeg_{1-2^{-t}}(F) \geq d_3$ (see the Remark following the proof of Theorem \ref{thm:big}). For the case of depth $d = 3$, we may use $\psi_3$ as a black box. However, to enable induction in the case of general $d$, we recall that the dual witness $\psi_3$ was defined via:

\begin{equation*} \psi_3(b_1, \dots, b_{t}) := 2^{t} \psi_2(\dots, \sgn(\psi_1(b_i)), \dots) \prod_{i=1}^{M} |\psi_1(b_i)|,\end{equation*}
 where $b_i = (b_{i, 1}, \dots, b_{i, M})$,
 $\psi_1$ was a dual witness to the high one-sided approximate degree of 
$\AND_M$, and $\psi_2$ was defined such that $\psi_2(\mathbf{1})=1/2$, $\psi_2(\mathbf{-1})=-1/2$, and $\psi_2$ evaluates to 0 for all other inputs in $\{-1, 1\}^{t}.$

By Theorem \ref{thm:oprelim} the dual witness $\psi_3$ satisfies:

\begin{equation} \label{eq:andorproof3} \sum_{b \in \{-1, 1\}^{t\cdot M}} \psi_3(b)F(b) > 1-2^{-t} = 1-n^{-2}, \end{equation}
\begin{equation} \label{eq:andorproof4} \sum_{b \in \{-1, 1\}^{t \cdot M}} |\psi_3(b)| = 1,\end{equation} 
\begin{equation} \label{eq:andorproof5} \sum_{b \in \{-1, 1\}^{t\cdot M}} \psi_3(b) \chi_S(b)=0   \text{ for each } |S| \leq d_3 \text{ and } \end{equation}
\begin{equation} \label{eq:andorproof6} \psi_3(b) \le 0   \text{ for each } b \in F^{-1}(-1).\end{equation}

Now let $\psi_4$ denote a dual witness to the fact that $\adeg_{.99}(\OR_{M/t})=\Omega(\sqrt{M/t})$. By Fact \ref{fact:odegand}, this dual witness has one-sided error, but on the side opposite from the one we used to define $\odeg$.
Thus there is some $d_4 = \Omega(\sqrt{M/t})$ such that the following equations hold:

\begin{equation} \label{eq:andorproof7} \sum_{w \in \{-1, 1\}^{M/t}} \psi_4(w)\OR_{M/t}(w) > .99, \end{equation}
\begin{equation} \label{eq:andorproof8} \sum_{w \in \{-1, 1\}^{M/t}} |\psi_4(w)| = 1,\end{equation} 
\begin{equation} \label{eq:andorproof9} \sum_{w \in \{-1, 1\}^{M/t}} \psi_4(w) \chi_S(w)=0   \text{ for each } |S| \leq d_4 \text{ and } \end{equation}
\begin{equation} \label{eq:andorproof10} \psi_4(\mathbf{1}) \ge 0.\end{equation}

Finally, we combine the dual witnesses $\psi_4$ and $\psi_3$ to obtain the desired function $\psi_5$:
\begin{equation} \label{eq:psiprime} \psi_5(z_1, \dots, z_{M/t}) := 2^{M/t} \psi_4(\dots, \sgn(\psi_3(z_i)), \dots) \prod_{i=1}^{M/t} |\psi_3(z_i)|,\end{equation}
where $z_i = (z_{i, 1}, \dots, z_{i, t \cdot M})$. 

\paragraph{Analysis of $\psi_5$.} 
The analysis in \cite{bunthaler}
immediately implies that $\psi_5$ has $L_1$-norm equal to 1, has pure high degree 
at least $d_3 \cdot d_4 = \Omega\left(M/\sqrt{t}\right)=\Omega\left(M/\sqrt{\log n}\right)$, and
that the correlation of $\psi_5$ with $\ORAND_{2, M^2}$ is at least $.99-2^{-t}\geq .98$.
What remains is to show that $\psi_5$ has ``almost no error on the wrong side''.
Recall that $A_{-1}=\{z \in \{-1, 1\}^{M^2} : \psi_5(z) < 0,  \ORAND_{2, M^2}(z)=1\}.$
We will show that:

\begin{equation} \label{eq:baderror} \sum_{z \in A_{-1}} |\psi_5(z)| \leq n^{-2}. \end{equation}

To establish \eqref{eq:baderror}, we first collect some observations. Let $B_{-1} = \{z_i \in \{-1, 1\}^{M\cdot t}: \psi_3(z_i) < 0, F(z_i) = 1\}$.  

\begin{itemize}
\item \noindent Observation 1:
For every $z=(z_1, \dots, z_{M/t}) \in \left(\{-1, 1\}^{t\cdot M}\right)^{M/t}$ 
in $A_{-1}$, the following property must hold:
$z_i \in B_{-1}$ for every $i$ such that $\psi_3(z_i)<0$.
This holds because $F(z_i)=1$ for \emph{all} $i \in \{1, \dots, M/t\}$, since
$\ORAND_{2, M^2}(z)=1$.

\item \noindent Observation 2: For every $z=(z_1, \dots, z_{M/t}) \in \left(\{-1, 1\}^{t\cdot M}\right)^{M/t} \in A_{-1}$, there must exist a $z_i$ such that $\psi_3(z_i)<0$.
This is because, if $\psi_3(z_i) \geq 0$ for all $i \in \{1, \dots, M/t\}$, then
$\psi_5(z)$ agrees in sign with $\psi_4(\mathbf{1}) > 0$ (see \eqref{eq:andorproof10}), contradicting the assumption
that $z \in A_{-1}$. 

\item \noindent Observation 3: Let $\mu$ be the distribution on $\{-1, 1\}^{M^2}$ 
defined via: $\mu(z_1, \dots, z_{M/t}) = \prod_{i=1}^{M/t} |\psi_3(z_i)|$. 
Since $\psi_3$ is balanced, the string $(\dots, \sgn(\psi_3(z_i)), \dots)$ is distributed uniformly in $\{-1, 1\}^{M/t}$
when one samples $z=(z_1, \dots, z_{M/t})$ according to $\mu$.  

\item \noindent Observation 4: Because
 $\psi_3$ has correlation $1-n^{-2}$ with $F$ (see \eqref{eq:andorproof3}),
 the following equation holds:
\begin{equation*}
\sum_{z_i \in B_{-1}} |\psi_3(z_i)| \leq \frac{1}{2}n^{-2}.\end{equation*}

\item \noindent Observation 5: As in the proof of Theorem \ref{thm:big},
let $\mu(z|w)$ denote the probability of $z$ under $\mu$, conditioned on $(\dots, \sgn(\psi_3(z_i)), \dots)=w$.
If $z \sim \mu(\cdot | w)$ for some string $w$ where $w_i = -1$, then the probability that $F(z_i) = 1$ when $\sgn(\psi_3(z_i)) = w_i$ is $2\sum_{z_i \in B_{-1}} |\psi_3(z_i)|$.

\eat{ the following two random variables are identically distributed:
\begin{itemize}
\item The string $(\dots, F(z_i), \dots)$ when one chooses $(\dots, z_i, \dots)$ from the conditional distribution 
$\mu(\cdot|w)$. 
\item The string $(\dots, y_iw_i, \dots)$, where $y \in \{-1, 1\}^{M}$ is a random string whose $i$th bit independently
takes on value $-1$ with probability $2 \sum_{z_i \in B_{w_i}} |\psi_3(z_i)|$.
\end{itemize}
}
\end{itemize}

Thus, we may write: 

\begin{align*}  \sum_{z \in A_{-1}} |\psi_5(z)| = \sum_{z \in A_{-1}} 2^{M/t} |\psi_4(\dots, \sgn(\psi_3(z_i)), \dots)| \prod_i |\psi_3(z_i)|   \end{align*}
\begin{align*} \leq \sum_{w \in \{-1, 1\}^{M/t}, w \neq \mathbf{1}} |\psi_4(w)|  \cdot \Pr_{z \sim \mu(\cdot|w)}[z_i \in B_{-1} \ \forall i: w_i=-1] \end{align*}
\begin{align*} \leq  \sum_{w \in \{-1, 1\}^{M/t}} |\psi_4(w)|  \cdot n^{-2} \leq n^{-2}.\end{align*}

Here, the equality holds by definition of $\psi_5$ (see \eqref{eq:psiprime}),
the first inequality holds by Observations 1, 2 and 3,
the second inequality holds by Observations 4 and 5, and the fourth inequality holds  because the $L_1$
norm of $\psi_4$ is 1 (see \eqref{eq:andorproof8}).

\end{proof}

Recall that we defined the combined dual witness
\[\psi_7(x_1, \dots, x_{M}) := 2^{M} \psi_6(\dots, \sgn(\psi_5(x_i)), \dots) \prod_{i=1}^{M} |\psi_5(x_i)|,\]
where $\psi_6$ is a dual polynomial for the high one-sided approximate degree of the top $\AND_M$ function. In the remainder of the proof, we show that $\psi_7$ is a dual witness for $\ANDOR_{3,n}$.

\paragraph{Bounding the Correlation of $\psi_7$ with $\ANDOR_{3, n}$.}
Using Equation \eqref{eq:baderror}, it is possible to adapt
the analysis of \cite{bunthaler} to show
\begin{claim} \label{claim:psi7}
$$\sum_{x} \psi_7(x) \ANDOR_{3, n}(x) > .97.$$
\end{claim}
\begin{proof}[Proof of Claim \ref{claim:psi7}]
The idea is to show that 
\begin{equation} \label{fuckingend} \sum_{x} \psi_7(x) \ANDOR_{3, n}(x) \approx \sum_{a \in \{-1, 1\}^M} \psi_6(a) \AND_{M}(a) > .99.\end{equation}
To this end,
let $A_{-1} = \{z \in \{-1, 1\}^{M^2}: \psi_5(z) < 0, \ORAND_{2, M^2}(z) = 1\}$ as above, 
and let $A_{1} = \{z \in \{-1, 1\}^{2, M^2}: \psi_5(z) \geq 0, \ORAND_{2, M^2}(z) = -1\}$.
Notice that
$A_1 \cup A_{-1}$ is the set of all inputs $z$ where the sign of $\psi_5(z)$ disagrees with $\ORAND_{2, M^2}(z)$. Notice that 
$\sum_{z \in A_1 \cup A_{-1}} |\psi_
5(z)| \leq .01$ because $\psi_5$ has correlation at least $.98$ with $\ORAND_{2, M^2}$. 

Let $\nu$ be the distribution on $\left(\{-1, 1\}^{M^2}\right)^M$ given by $\nu(x_1, \dots, x_M) = \prod_{i=1}^{M} |\nu(x_i)|$. Since $\nu$ is orthogonal to the constant polynomial, it has expected value 0, and hence the string $(\dots, \sgn(\psi_5(x_i)), \dots)$ is distributed uniformly in $\{-1, 1\}^{M}$
when one samples $(x_1, \dots, x_{M})$ according to $\nu$. 
Let $\nu(x_i|a)$ denote the probability of $x_i$ under $\nu$, conditioned on $(\dots, \sgn(\psi_5(x_i)), \dots)=a$.

For any given $a \in \{-1, 1\}^{M}$, the following two random variables are identically distributed:
\begin{itemize}
\item The string $(\dots, \ORAND_{2, M^2}(x_i), \dots)$ when one chooses $(\dots, x_i, \dots)$ from the conditional distribution 
$\nu(\cdot|a)$. 
\item The string $(\dots, y_ia_i, \dots)$, where $y \in \{-1, 1\}^{M}$ is a random string whose $i$th bit independently
takes on value $-1$ with probability $2 \sum_{x_i \in A_{a_i}} |\nu(x_i)| \le .02$. 
\end{itemize}

Thus, the left hand side of Expression (\ref{fuckingend}) equals

\begin{equation} \label{fuckingend2} \sum_{a \in \{-1, 1\}^{M}} \psi_6(a) \cdot \mathbf{E}[\text{AND}_M(\dots, y_ia_i, \dots)],\end{equation}

where $y \in \{-1, 1\}^{M}$ is a random string whose $i$th bit independently
takes on value $-1$ with probability $2 \sum_{x_i \in A_{a_i}} |\psi(x_i)| \le .02$. 

All $a \neq -\mathbf{1}_{M}$ can be handled exactly as in \cite{bunthaler} and \cite{sherstovFOCS}
to argue that they contribute at least $(1-.02)\psi_6(a)$ to the sum. The key
property exploited here is that $\AND_M$ has low \emph{block-sensitivity}
on these points, allowing us to apply the following proposition.

\begin{proposition}[\cite{sherstovFOCS}]\label{finalprop} Let $f: \{-1, 1\}^{M} \rightarrow \{-1, 1\}$ be a given Boolean function.
Let $y \in \{-1, 1\}^{M}$ be a random string whose $i$th bit is set to $-1$ with probability at most $\gamma \in [0, 1]$, 
and to $+1$ otherwise, independently for each $i$. Then for every $a \in \{-1, 1\}^{M}$, 
$$\mathbf{P}_y[f(a_1, \dots, a_{M}) \ne f(a_1y_1, \dots, a_Ma_M)] \leq 2\gamma\operatorname{bs}_a(f).$$ 
\end{proposition}

In particular, since $\text{bs}_a(\text{AND}_M) = 1$ for all $a \neq - \mathbf{1}_M$, Proposition \ref{finalprop} implies that for all $a \neq -\mathbf{1}_{M}$,
and $a=\text{AND}_{M}$, $\mathbf{P}_y[f(a_1, \dots, a_{M})= f(a_1y_1, \dots, a_My_M)] \geq 1-.02$.

We next argue that the term corresponding to $a=-\mathbf{1}_{M}$ contributes at least $(1-2Mn^{-2}) \psi_6(a)$ to Expression (\ref{fuckingend2}).
By \eqref{eq:baderror} and a union bound, for $a=-\mathbf{1}_{M}$, the $y_i$'s are \emph{all} $1$ with probability $1-2Mn^{-2}$, and hence $\mathbf{E}_y[\text{AND}_M\left(\dots, y_ia_i, \dots\right)] \le (1-2Mn^{-2})\text{AND}_M\left(-\mathbf{1}_{M}\right) = -(1-2Mn^{-2})$.
By \eqref{eq:andorproof100}, $\sgn(\psi_6(-\mathbf{1}_{M})) = -1,$ and thus the term corresponding to $a=-\mathbf{1}_{M}$ contributes at least $(1-2Mn^{-2})\psi_6(a)$ to Expression (\ref{eq2andor}) as claimed.  We conclude that $\sum_x \psi_7(x) \ANDOR_{3, n} \geq .97$.
\end{proof}

\eat{We remark that the preceding paragraph in fact establishes the following
fact, which will be important when extending the above argument to $\ANDOR_{d, n}$
for an arbitrary depth $d > 0$. Let $C_{1}=\{x: \psi_7(x) > 0 \cap \ANDOR_{3, n}(x) = 1\}$.
Then $W:=\sum_{x \in C_1} |\psi_7(x)| \leq 1-2Mn^{-2}$. 
the quantity $W$ should be thought of as the error of $\psi_7$ ``on the wrong side'',
and it is important when extending to larger depths $d$ that $W$ be small.}

\medskip
\noindent \textbf{Completing the proof for $d=3$.}
The proof that $\psi_7$ has $L_1$-norm 1 and has pure high degree at least
$d_5 \cdot d_6 = \Omega\left(n^{1/2}/\log^{1/2}(n)\right)$ is identical
to prior work \cite{sherstovFOCS} (see also Appendix \ref{app:proof1}).
Combined with Claim \ref{claim:psi7} showing that $\sum_x \psi_7(x) \ANDOR_{3, n} \geq .97$,
we conclude that $\psi_7$ is a dual witness to the fact that 
$\adeg_{.97}(\ANDOR_{3, n})=\Omega\left(n^{1/2}/\log^{1/2}(n)\right)$. 

\medskip
\noindent \textbf{Extending to general $d$.}
For ease of exposition, we focus on the case where $d$ is odd;
the case of even $d$ is similar. To enable a proof by induction, we will show that

\begin{claim}\label{claim:induction}
For $d$ odd, there exists a dual witness $\psi_5$ showing that the $.99$-approximate degree of $\ANDOR_{d, n}$ is $\Omega(n^{1/2}/\log^{(d-1)/2}(n))$. Moreover,
\begin{equation}  \label{eq:minbaderror} \sum_{y \in A_1} |\psi_5(y)| \leq 2n^{-2},\end{equation}
where $A_{1} = \{y: \psi_5(y) > 0, \ANDOR_{d, n}(y) = -1\}$.

For $d$ even, the same statement holds for the approximate degree of $\ORAND_{d, n}$ where we replace \eqref{eq:minbaderror}
with the corresponding bound on $\sum_{y \in A_{-1}} |\psi_5(y)|$, 
where $A_{-1} = \{y:\psi_5(y) < 0, \ORAND_{d, n}(y) = 1\}$.
\end{claim}

\eqref{eq:minbaderror} intuitively captures the property that
$\psi_1'$ has ``almost no error on the wrong side''. 

\begin{proof}[Proof of Claim \ref{claim:induction}]
As a base case of the induction, the dual witness $\psi_1$ 
that we used in the case $d=3$ clearly satisfies the above 
properties (in fact, $\psi_1$ had one-sided error, and therefore
satisfied an even stronger  condition than \eqref{eq:minbaderror}).

As suggested by our choice of $\psi_5$ as the name of the function we want to construct, the inductive case mimics the proof of Claim \ref{claim:psi5}. To emphasize the similarity between this argument and the proof of Claim \ref{claim:psi5}, we will show that assuming the induction hypothesis at level $d-2$ implies the induction hypothesis at level $d-1$, for $d$ odd. The case of $d$ even is similar.

To construct a dual witness $\psi_5$ proving that $\adeg(\ORAND_{d-1, n^{1-1/d}}) = \Omega(n^{(1-1/d)/2}/\log^{(d-2)/2}(n))$ with ``almost no error on the wrong side,'' we inductively assume
that there exists a dual witness $\psi'_1$ for the high approximate degree of the function $G=\ANDOR_{d-2, n^{1-2/d}}$ with almost no error on the wrong side.
That is, there exists a $\psi_1'$ and a $d_1' = \Omega(n^{(1-2/d)/2}/\log^{(d-3)/2}(n))$ such that

\begin{equation} \label{eq:andorproofgend0} \sum_{y \in \{-1, 1\}^{n^{1-2/d}}} \psi'_1(y)G(y) > .99, \end{equation}
\begin{equation} \label{eq:andorproofgend1} \sum_{y \in \{-1, 1\}^{n^{1-2/d}}} |\psi'_1(y)| = 1, \text{ and} \end{equation} 
\begin{equation} \label{eq:andorproofgend2} \sum_{y \in \{-1, 1\}^{n^{1-2/d}}} \psi'_1(y) \chi_S(y)=0   \text{ for each } |S| \leq d_1',\end{equation}
in addition to having ``almost no error on the wrong side.''

Now we set $M=n^{1/d}$, and define $\psi_2, \psi_3, \dots, \psi_5$ exactly as in the case $d=3$,
but with the dual witness $\psi'_1$ in place of the dual witness
$\psi_1$.
That is, we let $\psi_2:\{-1, 1\}^t \rightarrow \mathbb{R}$ be defined via $\psi_2(\mathbf{1})=1/2$, $\psi_2(\mathbf{-1})=-1/2$, and $\psi_2(b_i)=0$ for all other $b_i \in \{-1, 1\}^{t}.$
We define \begin{equation*}\psi_3(b_1, \dots, b_{t}) := 2^{t} \psi_2(\dots, \sgn(\psi'_1(b_i)), \dots) \prod_{i=1}^{M} |\psi'_1(b_i)|,\end{equation*}
 where $b_i = (b_{i, 1}, \dots, x_{i, M})$.
 We define $\psi_4$ to be a dual witness
 to the fact that $\adeg_{.99}(\OR_{M/t})=\Omega(\sqrt{M/t})$ for $t = 2\log n$. 
 We define $\psi_5$ exactly as in \eqref{eq:psiprime}.

The analysis of $\psi_5$ proceeds as in the proof of Claim \ref{claim:psi5}, with one modification. 
In the case of $d=3$, $\psi_1$ had one-sided error, 
so we could directly invoke our hardness amplification
result (Theorem \ref{thm:big}) to conclude that 
$\psi_3$
also had one-sided error, as well as correlation $1-2^{-t}$ with the target
function $\OR_t(G, \dots, G)$. 
In the case of general $d$, $\psi'_1$ does not have one-sided error.
However, $\psi_1'$ ``almost'' has one-sided error, 
as formalized by \eqref{eq:minbaderror}. 
It is straightforward to modify the proof of Theorem \ref{thm:big}
to show though $\psi'_1$ satisfies a weaker condition than did
$\psi_1$,
the dual witness
$\psi_3$ nonetheless
satisfies the following properties.

Let $B_{-1} = \{z_i \in \{-1, 1\}^{n^{1-2/d}\cdot t}: \psi_3(z_i) < 0, \OR_t(G, \dots, G)(z_i) = 1\}$,
and let $B_{1} = \{z_i \in \{-1, 1\}^{n^{1-2/d}\cdot t}: \psi_3(z_i) > 0, \OR_t(G, \dots, G)(z_i) = -1\}$.
Then:
\begin{itemize}
\item $\sum_{z_i \in B_{-1}} |\psi_3(z_i)| \leq 2^{-t}$.
\item $\sum_{z_i \in B_{1}} |\psi_3(z_i)| \leq t \cdot 2n^{1-2/d}/n^{2}$.
\end{itemize}

That is, $\psi_3$ has error exponentially small in $t$ on one side, and the error
on the other side blows up by at most a factor of $t$ relative to $\psi_1$. This 
permits us to obtain a variant of \eqref{eq:baderror}, namely:

\begin{equation} \label{eq:inductive} \sum_{z \in A_{-1}}|\psi_5(z)| \leq 2tM/n^{2},\end{equation}
where as above $A_{-1}$ is defined via:
$$A_{-1} = \{z \in \{-1, 1\}^{n^{1-1/d}}: \psi_5(z) < 0, \ORAND_{d-1, n^{1-1/d}}(z) = 1\}.$$
This completes the induction and the proof.

\end{proof}

With Claim \ref{claim:induction} in hand, we can construct $\psi_7$ as in the proof of the $d = 3$ case to obtain Theorem \ref{thm:andor}. That is, we define $\psi_6$ to be a dual witness to the high one-sided approximate
 degree of $\AND_M$, and we define $\psi_7$ exactly as in 
 \eqref{eq:zetaprimerepeat}.
 
 As before, $\psi_7$ has $L_1$-norm 1 and pure high degree at least
$d_5 \cdot d_6 = \Omega\left(n^{1/2}/\log^{(d-2)/2}(n)\right)$.
Here, $d_5=\Omega\left(n^{(1-1/d)/2}/\log^{(d-2)/2}(n)\right)$
denotes the pure high degree of $\psi_5$ and $d_6=\Omega\left(M^{1/2}\right)$ denotes
the pure high degree of $\psi_6$. Finally, the analysis establishing that $\psi_7$ has high correlation with $\ANDOR_{d, n}$ is the same as in the case of $d=3$.

 \end{proof}

\section{Lower Bounds for Read-Once DNFs}
\label{sec:readonce}
In this section we derive new approximate degree and degree-$d$ threshold weight lower bounds for read-once DNF formulas. The lower bounds we prove are essentially identical to those proved by Beigel \cite{beigel93} and Servedio et al. \cite{thalercolt} for the \emph{decision list} \omb, which is not computable by a read-once DNF. Our first construction (Corollary \ref{cor:colt}) yields a degree-$d$ threshold weight lower bound of $2^{\Omega(\sqrt{n/d})}$, matching the lower bound proved by Servedio et al. for the decision list \omb. In \secref{sec:dnf-tightness}, we show that this is essentially optimal in the ``high-degree'' regime where $d = \Omega(n^{1/3})$.

Our second lower bound (Corollary \ref{cor:DNF-adeg}) exhibits a DNF with $(1-2^{-n/d^2})$-approximate degree $\Omega(d)$, matching Beigel's lower bound
for \omb. 
As we remarked in Section \ref{sec:readonceintro}, for $d < n^{1/3}$, 
Corollary \ref{cor:DNF-adeg} is subsumed
by Minsky and Papert's seminal result exhibiting a read-once
DNF $F$ with threshold degree $\Omega(n^{1/3})$.
However, for $d > n^{1/3}$, it is not subsumed
by Minsky and Papert's result, nor by
Corollary \ref{cor:colt}. While
Corollary \ref{cor:colt} yields a lower bound on the degree-$d$
threshold weight of read-once DNFs, it does not yield a lower bound
on the \emph{approximate-degree} of read-once DNFs.
As described in Section \ref{sec:weightrelations},
while $\adeg_{1 - \frac{1}{w}}(F) > d$
implies that $W(F, d) > w$,
the reverse implication does \emph{not} hold when $w \ll {n \choose d}$
(and in fact the read-once DNF considered in Corollary \ref{cor:colt}
is an explicit example of the reverse implication failing badly).

\subsection{Extending the Lower Bound of Servedio et al.
to Read-Once DNFs}

\subsubsection{Hardness Amplification for Approximate Weight}
We now extend our hardness amplification techniques
from approximate degree
to approximate weight.
This extension forms the technical heart of our proof
that the lower bound of Servedio et al. applies to read-once DNFs.

\begin{customthm}{\ref{thm:big-weights}}
Let $f: \{-1, 1\}^m \to \{-1, 1\}$ be a function with one-sided non-constant approximate weight $W_{3/4}^*(f, d) > w$. Let $F: \{-1, 1\}^{mt} \to \{-1, 1\}$ denote the function $\OR_{t}(f, \dots, f)$. Then $F$ has degree-$d$ $(1-2^{-t})$-approximate weight $W_{1-2^{-t}}(F, d) > 2^{-5t}w$.
\end{customthm}

\begin{proof} 
Let $\psi$ be a dual polynomial for $f$ with one-sided error whose existence is guaranteed by the assumption that $W_{3/4}^*(f, d) > w$. 
Then by \thmref{thm:owprelim}, $\psi$ satisfies:

\begin{equation} \label{eq:wproof0} \sum_{x \in \{-1, 1\}^m} \psi(x)f(x) - \frac{3}{4}\sum_{x \in \{-1, 1\}^m} |\psi(x)| > w, \end{equation}
\begin{equation} \label{eq:wproof1} \left| \sum_{x \in \{-1, 1\}^m} \psi(x) \chi_S(x) \right| \le 1   \text{ for each } 0 < |S| \leq d, \end{equation}
\begin{equation} \label{eq:wproof2} \sum_{x \in \{-1, 1\}^m} \psi(x) = 0,
 \text{ and } \end{equation}
\begin{equation} \label{eq:wproof3} \psi(x) \le 0   \text{ for each } x \in f^{-1}(-1).\end{equation}

We will construct a dual solution $\zeta$ that witnesses the fact that $W_{1-2^{-t}}(F, d) > 2^{-5t}w$. Specifically, by \thmref{thm:wprelim}, $\zeta$ must satisfy the following conditions:

\begin{equation} \label{eq:wshow1} \sum_{(x_1, \dots, x_{t}) \in \left(\{-1, 1\}^{m}\right)^t} \zeta(x_1, \dots, x_{t}) F(x_1, \dots, x_{t}) - (1-2^{-t})|\zeta(x_1, \dots, x_{t})| > 2^{-5t}w. \end{equation}
\begin{equation} \label{eq:wshow2}  \left| \sum_{(x_1, \dots, x_{t}) \in \left(\{-1, 1\}^{m}\right)^t}  \zeta(x_1, \dots, x_{t})\chi_S(x_1, \dots, x_t ) \right| \le 1   \text{ for each } |S| \leq d.\end{equation}

As before, let $\Psi : \{-1, 1\}^t \rightarrow \{-1, 1\}$ be defined such that $\Psi(\mathbf{1})=1/2$, $\Psi(-\mathbf{1})=-1/2$, and $\Psi(x)=0$ for all other $x$, where $\mathbf{1}$ denotes the all-ones vector. We define $\zeta: \left(\{-1, 1\}^{m}\right)^{t} \rightarrow \mathbb{R}$ by
\begin{equation} \label{eq:wzeta} \zeta(x_1, \dots, x_{t}) := M_t \Psi(\dots, \sgn(\psi(x_i)), \dots) \prod_{i=1}^{t} |\psi(x_i)|,\end{equation}
where $x_i = (x_{i, 1}, \dots, x_{i, m})$ and $M_t$ is a normalization term to be determined later.

We start with \eqref{eq:wshow2} to determine an appropriate choice of $M_t$. Notice that since $\Psi$ is orthogonal on $\{-1, 1\}^t$ to constant functions, its expected value is 0. Thus, we may write the Fourier representation for $\Psi$ as
$$\Psi(z) =	\sum_{\substack{T \subseteq \{1, \dots, t\} \\ T \neq \emptyset}}	\hat{\Psi}(T) \chi_T(z)$$
for some real numbers $\hat{\Psi}(T)$. We can thus write
\[\zeta(x_1, \dots, x_t) = M_t \sum_{T \neq \emptyset} \hat{\Psi}(T) \prod_{i \in T} \psi(x_i)\prod_{i \notin T} |\psi(x_i)|.\]
Given a subset $S \subseteq \{1, \dots, t\} \times \{1, \dots, m\}$ with $|S| \le d$, partition $S = (\{1\} \times S_1) \cup \dots \cup (\{t\} \times S_t)$ where each $S_i \subseteq \{1, \dots, m\}$. Then
\begin{align*}
\sum_{(x_1, \dots, x_{t}) \in \left(\{-1, 1\}^{m}\right)^t}& \zeta(x_1, \dots, x_{t})\chi_S(x_1, \dots, x_t) \\
&= M_t \sum_{T \neq \emptyset} \hat{\Psi}(T) \prod_{i \in T} \underbrace{\left( \sum_{x_i \in \{-1, 1\}^m} \psi(x_i) \chi_{S_i}(x_i)\right)} \prod_{i \notin T} \left(\sum_{x_i \in \{-1, 1\}^m} |\psi(x_i)| \chi_{S_i}(x_i)\right).
\end{align*}
Since $|S| \le d$, we have that $|S_i| \le d$ for every index $i \in \{1, \dots, t\}$. For each set $T$, each of the underbraced factors is bounded in absolute value by $1$ by (\ref{eq:wproof1}). Writing
\[\|\psi\|_1 := \sum_{x \in \bits^m} |\psi(x)|\]
for notational convenience, we see that
\[\left|\sum_{(x_1, \dots, x_{t}) \in \left(\{-1, 1\}^{m}\right)^t} \zeta(x_1, \dots, x_{t})\chi_S(x_1, \dots, x_t)\right| \le M_t \sum_{T \ne \emptyset} \hat{\Psi}(T) \|\psi\|_1^{t - |T|} \le M_t\cdot t2^{t-1}\|\psi\|_1^{t-1}.\]
Taking $M_t = 2^{-2t}\|\psi\|_1^{1-t}$ gives (\ref{eq:wshow2}).

We now proceed to verify (\ref{eq:wshow1}). Let $\mu$ be the distribution on $\left(\{-1, 1\}^{m}\right)^t$ given by $\mu(x_1, \dots, x_t) = \|\psi\|_1^{-t}\prod_{i=1}^{t} |\psi(x_i)|$. 
Since $\psi$ is orthogonal to the constant polynomial, it has expected value 0, and hence the string $(\dots, \sgn(\psi(x_i)), \dots)$ is distributed uniformly in $\{-1, 1\}^{t}$
when one samples $(x_1, \dots, x_{t})$ according to $\mu$. 
 Observe that
$$ \sum_{(x_1, \dots, x_{t}) \in \left(\{-1, 1\}^{m}\right)^t} \zeta(x_1, \dots, x_{t}) F(x_1, \dots, x_{t})$$
$$= M_t \|\psi\|_1^t\mathbf{E}_\mu [\Psi( \dots, \sgn(\psi(x_i)), \dots) \OR_t\left( \dots, f(x_i), \dots\right)]$$
\begin{equation} \label{weq1andor} = 2^{-3t}\|\psi\|_1 \sum_{z \in \{-1, 1\}^{t}} \Psi(z) \left(\sum_{(x_1, \dots, x_{t}) \in \left(\{-1, 1\}^m\right)^t} \OR_t\left( \dots,  f(x_i), \dots\right) \mu(x_1, \dots, x_{t}|z)\right),\end{equation}
where $\mu(\mathbf{x}|z)$ denotes the probability of $\mathbf{x}$ under $\mu$, conditioned on $(\dots, \sgn(\psi(x_i)), \dots)=z$.

Let $A_{1} = \{x \in \{-1, 1\}^{m}: \psi(x) > 0, f(x) = -1\}$ and $A_{-1} = \{x \in \{-1, 1\}^{m}: \psi(x) < 0, f(x) = 1\}$. Then 
$2\sum_{x \in A_1 \cup A_{-1}} |\psi(x)| < \frac{1}{4} \|\psi\|_1 - w$ because $\psi$ has correlation at least $w + \frac{3}{4} \|\psi\|_1$ with $f$. 

As before, for any $z \in \{-1, 1\}^{t}$, the following two random variables are identically distributed:

\begin{itemize}
\item The string $(\dots, f(x_i), \dots)$ when one chooses $(\dots, x_i, \dots)$ from the conditional distribution 
$\mu(\cdot|z)$. 
\item The string $(\dots, y_iz_i, \dots)$, where $y \in \{-1, 1\}^{t}$ is a random string whose $i$th bit independently
takes on value $-1$ with probability $\frac{2}{\|\psi\|_1} \sum_{x \in A_{z_i}} |\psi(x)| < 1/4 - w/\|\psi\|_1$. 
\end{itemize}

Thus, the correlation is

\begin{equation} \label{weq2andor} 2^{-3t}\|\psi\|_1 \sum_{z \in \{-1, 1\}^{t}} \Psi(z) \cdot \mathbf{E}[\OR_t(\dots, y_iz_i, \dots)],\end{equation}

where $y \in \{-1, 1\}^{t}$ is a random string whose $i$th bit independently
takes on value $-1$ with probability $2 \sum_{x \in A_{z_i}} |\psi(x)| < 1/4 - w/\|\psi\|_1$.
As in the proof of \thmref{thm:big}, the one-sided error (\ref{eq:wproof3}) of the dual witness $\psi$ implies that the input $z=\mathbf{1}$ contributes $\Psi(z) = 1/2$ to Expression (\ref{weq2andor}).
All $z \not\in \{\mathbf{1}, \mathbf{-1}\}$ are given zero weight by $\Psi$ and hence contribute nothing to the sum. 
All that remains is to show that the contribution of the term $z=-\mathbf{1}$ to the sum is $\frac{1}{2} (1- 2^{-2t+1})$.
Since each $y_i=1$ independently with probability at least $3/4 + w/\|\psi\|_1$, and $\OR_t(\dots, -y_i, \dots)=-1$ as long as there is at least one $y_i \neq -1$,
we conclude that  $\mathbf{E}[\OR_t(\dots, y_iz_i, \dots)] \geq 1-2^{-2t+1}$. It follows that  the term corresponding to
$z=-\mathbf{1}$ contributes at least $\frac{1}{2}(1 - 2^{-2t+1})$ to the sum. Thus,
$$2^{-3t}\|\psi\|_1 \sum_{z \in \{-1, 1\}^{t}} \Psi(z) \cdot \mathbf{E}[\OR_t(\dots, y_iz_i, \dots)] \geq 2^{-3t}\|\psi\|_1 \left( \frac{1}{2} + \frac{1}{2} (1-2^{-2t+1}) \right) = 2^{-3t}(1-2^{-2t})\|\psi\|_1.$$

Since $\psi$ is orthogonal to the constant polynomial by \eqref{eq:wproof2}, it has expected value 0, and hence the string $(\dots, \sgn(\psi(x_i)), \dots)$ is distributed uniformly in $\{-1, 1\}^{t}$
when one samples $(x_1, \dots, x_t)$ according to $\mu$. 
Thus,
$$ \sum_{(x_1, \dots, x_t) \in \left(\{-1, 1\}^{m}\right)^t} |\zeta(x_1, \dots, x_t)| = 2^{-3t}\|\psi\|_1 \sum_{z \in \{-1, 1\}^t} |\Psi(z)| = 2^{-3t}\|\psi\|_1,$$.

Now the left-hand side of Expression (\ref{eq:wshow1}) is at least
\[2^{-3t}(1-2^{-2t})\|\psi\|_1 - (1-2^{-t}) \cdot 2^{-3t}\|\psi\|_1 > 2^{-5t}\|\psi\|_1 > 2^{-5t}w,\]
where the last inequality follows from condition (\ref{eq:wproof0}). This completes the proof.

\end{proof}

\subsubsection{Completing the Proof of Corollary \ref{cor:colt}}
We adapt an argument of Servedio et al. to prove the following
one-sided approximate weight lower bound for the function $\AND_n$.

\begin{lemma}\label{lem:and-approx-weight}
Let $d = o(n/\log^2 n)$. Then the function $\AND_n$ has one-sided non-constant approximate weight $W^*_{3/4}(\AND_n, d) = 2^{\Omega(n/d)}$.
\end{lemma}

Our proof of Lemma \ref{lem:and-approx-weight} follows a symmetrization argument due to Servedio et al. \cite{thalercolt}. The key in their proof is the following Markov-type inequality that gives a sharp bound on the derivative of a bounded polynomial in terms of both its degree and weight.

\begin{lemma}[\cite{thalercolt}, Lemma 1]\label{lem:weighted-markov}
Let $P: \R \to \R$ be a degree-$d$ polynomial such that
\begin{enumerate}
\item The coefficients of $P$ each have absolute value at most $w$, and
\item $1/2 \le \max_{x \in [-1, 1]} |p(x)| \le R$.
\end{enumerate}
Then $\max_{x \in [-1, 1]} |p'(x)| = O(d\cdot R\cdot\max\{\log W, \log d\})$.
\end{lemma}

\begin{proof}[Proof of \lemref{lem:and-approx-weight}]
Let $p: \R^n \to \R$ be a real polynomial with degree $d$ and non-constant weight $w$ that has one-sided distance at most $3/4$ from $\AND_n$. Specifically, $p(-\mathbf{1}) \le -1/4$ and $1/4 \le p(x) \le 7/4$ at all other Boolean inputs. We will show that $w = 2^{\Omega(n/d)}$. First observe that if $p(-\mathbf{1}) \le -7/4$, then the polynomial
\[q(x) = \frac{2(p(x) - 1)}{|p(-\mathbf{1}) - 1|} + 1\]
is a true $(3/4)$-approximation to $\AND_n$ with weight smaller than $w + 1$, so we can assume without loss of generality that $p$ is in fact a $(3/4)$-approximation to $\AND_n$.

Define the univariate polynomial
\[P(t) := \E_{x \gets \mu_t}[p(x)]\]
where $\mu_t$ is the product distribution over $\bits^n$ where each coordinate $x_j$ is independently set to $1$ with probability $(1+t)/2$. Notice that $P(t)$ is obtained from the multivariate expansion of $p(x_1, \dots, x_n)$ by replacing each variable $x_i$ with $t$. It is readily verified that $P$ satisfies the following properties.

\begin{enumerate}
\item $P(-1) = p(-\mathbf{1})$ and $P(1) = p(\mathbf{1})$,
\item $|P(t)| \le \frac{7}{4}$ for all $t \in [-1, 1]$, and
\item $\deg P \le \deg p = d$.
\item $P$ has non-constant weight at most $w$.
\end{enumerate}
By combining properties (1) and (4), we additionally see that the constant term $P(0)$ has absolute value at most $w + \frac{7}{4}$.  We can then verify that $P$ satisfies the conditions of \lemref{lem:weighted-markov}.
\begin{enumerate}
\item The coefficients of $P$ each have absolute value at most $w + \frac{7}{4}$ and
\item $1/2 \le \max_{x \in [-1, 1]}|P(t)| \le \frac{7}{4}$.
\end{enumerate}
Thus we conclude that $|P'(t)| = O(d\max\{\log w, \log d\})$ for $t \in [-1, 1]$. On the other hand, at $t_0 = -1 + 2/n$, we have $\Pr_{x \gets \mu_{t_0}}[x = -1^n] = (1 - \frac{1}{n})^n < 1/e$, so $P(t_0) \ge 1 - \frac{2}{e}$. Since $P(-1) = p(-\mathbf{1}) \le -\frac{1}{4}$, by the mean value theorem, there is some $t \in [-1, t_0]$ where $P'(t) \ge \frac{n}{4}$. Thus we have $d\max\{\log w, \log d\} = \Omega(n)$, and hence $w = 2^{\Omega(n/d)}$ as long as $d = o(n/\log^2 n)$.
\end{proof}

Finally, we are in a position to prove Corollary \ref{cor:colt}, restated
here for the reader's convenience.

\corcolt*

\begin{proof}
Set $m = \alpha\sqrt{nd}$ where $\alpha$ is a constant to be determined later, and let $t = n/m = \Omega(\sqrt{n/d})$. Let $F = \OR_t(\AND_m, \dots, \AND_m)$. By \lemref{lem:and-approx-weight}, the inner function $\AND_m$ has degree-$d$ one-sided non-constant approximate weight $W^*_{3/4}(\AND_m, d) = 2^{\beta m/d}$ for some constant $\beta$. Since $d = o(m/\log^2 m)$, by \thmref{thm:big-weights} the composed function $F$ has degree-$d$ approximate weight 
$$W_{1-2^{-t}}(F, d) = 2^{-5t + \beta m/d} = 2^{(-5/\alpha + \beta)\sqrt{n/d}}.$$
Setting $\alpha > 5/\beta$, we get that this approximate weight is greater than $1$. By \lemref{lem:weight-relations}, we have that $W(F, d) > 2^{-t} = 2^{\Omega(\sqrt{n/d})}$.
\end{proof}

\subsection{Extending Beigel's Lower Bound to Read-Once DNFs}

\corDNFadeg*

\begin{proof}
Let $m=d^2$, $t = n/d^2$, and $f=\AND_{m}$. Then Theorem \ref{thm:big} guarantees that $$\adeg_{1-2^{-t}}\left(\OR_{t}(\AND_m, \dots, \AND_m)\right) > \odeg(f).$$
By Fact \ref{fact:odegand}, the one-sided approximate degree of $f$ is $\Omega(\sqrt{m})$. This completes the proof.
\end{proof}

\eat{A threshold weight lower bound matching the original statement of Beigel's result for \omb\ follows immediately from Corollary \ref{cor:DNF-adeg}
and Lemma \ref{lem:weight-relations}.

\begin{corollary} \label{cor:beigel}
For each $d>0$, there is a read-once DNF $F$
satisfying $W(F, d) = \exp\left(\Omega\left(n/d^2\right)\right)$. 
\end{corollary}
}

\subsection{On the Tightness of Corollaries \ref{cor:colt} and \ref{cor:DNF-adeg}}\label{sec:dnf-tightness}

In \secref{sec:tightness}, we showed that \corref{cor:ed-amp} is essentially tight by exhibiting a nearly-matching upper bound
based on rational approximations. A similar construction
shows that any DNF of top fan-in $t$ is computed by a PTF
of degree $\tilde{O}(t)$ and weight $\exp\left(\tilde{O}(t)\right)$.
This construction immediately shows that Corollary \ref{cor:colt}
is tight (up to logarithmic factors) for all $d > n^{1/3}$. Indeed, the DNF $F$ 
for which Corollary \ref{cor:colt}
demonstrates $W(F, d) \geq \exp(\Omega(\sqrt{n/d}))$ has top fan-in
$t=\sqrt{n/d}$, which is less than $d$ for all $d > n^{1/3}$. 
This construction also reveals
a sharp thresholding phenomenon for the read-once DNFs 
considered in Corollaries \ref{cor:colt} and \ref{cor:DNF-adeg} that is similar to the one observed for the depth-three circuit
considered in \secref{sec:sharp-threshold}.

However, we can provide an alternative construction that demonstrates the tightness of both Corollaries \ref{cor:colt} and \ref{cor:DNF-adeg}. Specifically, rather than utilizing rational approximation techniques, we can construct a PTF for a read-once DNF by composing a PTF for the top $\OR$ gate with low-degree (polynomial, rather than rational) pointwise approximations to each of the individual terms. We provide this 
construction because of its power to explain why the lower bounds
of Corollaries \ref{cor:colt} and \ref{cor:DNF-adeg} take their particular forms.

Fix any function $f: \{-1, 1\}^m \rightarrow \{-1, 1\}$, and let $p: \{-1, 1\}^m \rightarrow \{-1, 1\}$ be a polynomial of degree $d$
and weight $w$ such that $|p(x) - f(x)| < 1/t$ for all $x \in \bits^m$. 
Let $F(x_1, \dots, x_t)=\OR_t(f(x_1), \dots, f(x_t))$.
Then for $(x_1, \dots, x_t) \in \{-1, 1\}^{m \cdot t}$, the identity $F(x_1, \dots x_t) = \operatorname{sgn}(1 - t + \sum_{i=1}^t p(x_i))$ yields a PTF for 
$F$ of degree at most $d$ and weight at most $tw + t+1$.

Recall that Corollary \ref{cor:colt}
yields a lower bound of $W(F, d) = \exp\left(\Omega(\sqrt{n/d})\right)$, where $F$ is the read-once DNF with top fan-in roughly $t=\sqrt{n/d}$ and bottom fan-in roughly $m=\sqrt{nd}$. Servedio et al. \cite{thalercolt} showed that for any $d > m^{1/2}$, there is a polynomial $p$
of degree $\tilde{O}(d)$ and weight $\exp\left(\tilde{O}(m/d + \log t)\right)=\exp\left(\tilde{O}\left(\sqrt{n/d}\right)\right)$
that approximates the function $\AND_m$ to error $1/t^2$. 
Hence, as long as $d > n^{1/3}$, the polynomial
$1 - t + \sum_{i=1}^t p(x_i)$ is a 
PTF for $F$ of degree $\tilde{O}(d)$ and weight $\exp\left(\tilde{O}(\sqrt{n/d})\right)$, showing that \corref{cor:colt} is tight up to logarithmic factors.

Similarly, recall that Corollary \ref{cor:DNF-adeg} yields a lower bound of $\adeg_{1-2^{n/d^2}}(F) = \Omega(d)$, where $F$ is the read-once DNF with top fan $t=n/d^2$ and bottom fan-in $m=d^2$. It is well-known that a transformation of the Chebyshev polynomials 
yields a polynomial $p$ of degree $\tilde{O}(m^{1/2})$ and weight $\exp\left(\tilde{O}(m^{1/2} + \log t)\right)$ that approximates $\AND_m$ to error better than
$1/t^2$ (see e.g. \cite{ksomb}). Hence, $1 - t + \sum_{i=1}^t p(x_i)$ is a 
PTF for $F$ of degree $\tilde{O}(d)$ and weight $\exp(\tilde{O}(d + \log t)) = \exp(\tilde{O}(n/d^2))$ when $d < n^{1/3}$. The transformation of \lemref{lem:weight-relations} then shows that \corref{cor:DNF-adeg} is tight up to logarithmic factors in this parameter range.

\section{Applications}
\label{sec:applications}
In this section, we detail applications of the results described above
to communication complexity, circuit complexity, and computational
learning theory.

\subsection{Communication Complexity}
Let $f: X \times Y \rightarrow \{-1, 1\}$, where $X$ and $Y$ are finite sets. Consider a two-party communication
problem in which
Alice is given an input $x \in X$, Bob is given an input $y \in Y$, and their goal is to compute
$f(x, y)$ with probability $1/2 + \beta$ for some bias $\beta > 0$. Alice and Bob each have access to an arbitrarily
long sequence of private random bits, and the cost $C(P)$ of a protocol $P$ is the worst-case number of 
bits they must exchange over all inputs $(x, y) \in X \times Y$. Babai et al. \cite{bfs} defined the \emph{PP communication} model to capture the complexity of computing $f$ with small bias. The PP communication
complexity of $f$, denoted by $\PP(f)$, is the minimum value of $C(P) + \log(1/\beta(P))$ over all protocols $P$ that compute $f$ with positive bias.

It is well known \cite{klauckquantum} that PP communication is essentially characterized by discrepancy: If $f: \bits^n \times \bits^n \to \bits$, then $\PP(f) = \Theta\left(\log \left(1/\disc(f)\right) + \log n\right)$.
It follows immediately that our $\exp\left(-\tilde{\Omega}(n^{2/5})\right)$ upper bound on
the discrepancy of an \acz\ function $f$ implies an
$\tilde{\Omega}(n^{2/5})$ lower bound on $\text{PP}(f)$. 
The previous best lower bound on $\text{PP}(f)$ for an \acz\ function
$f$ was $\Omega(n^{1/3})$ \cite{patternmatrix, bvdw}.

\subsection{Circuit Complexity}
Constant-depth circuits of majority gates are known to be surprisingly powerful. 
Most strikingly, Allender \cite{allender} showed that any function in \acz\ can be computed
by a depth three circuit of majority gates of quasipolynomial size. This
prompted Krause and Pudl{\'a}k \cite{krausepudlak} to ask whether 
every \acz\ function could be computed by depth \emph{two} majority 
gates of polynomial size. This question was resolved in the negative by Sherstov \cite{majmaj},
who exhibited an \acz\ function that cannot be computed even by majority-of-threshold circuits
of size $\exp(n^{1/5})$ (later sharpened to $\exp(n^{1/3})$ \cite{patternmatrix}), and
 independently by Buhrman, Vereshchagin, and de Wolf \cite{bvdw}, who obtained an $\exp(n^{1/3})$
lower bound on the size of majority-of-threshold circuits computing a different \acz\ function.

It is well-known that a discrepancy upper bound for $F$ yields a lower bound on the size of majority-of-threshold circuits
computing $F$ \cite{ghr, hajnal, nisan, majmaj}, and indeed, the circuit lower bounds 
of \cite{majmaj, patternmatrix,bvdw} are all proved using discrepancy. Through this connection, our discrepancy upper bound of Corollary \ref{cor:discrepancy} sharpens the previous lower bounds by yielding a depth-four Boolean circuit $F$ of polynomial
size
such that any majority-of-threshold circuit computing $F$ requires size $\exp\left(\tilde{\Omega}(n^{2/5})\right)$.

\begin{restatable}{corollary}{corcircuits}\label{cor:majthresh} There is a depth-four Boolean circuit $F: \{-1, 1\}^n \rightarrow \{-1, 1\}$ 
 of size $\poly(n)$ such that every majority-of-threshold circuit computing $F$ 
 has size $\exp\left(\tilde{\Omega}(n^{2/5})\right)$.
\end{restatable}

\subsection{Learning Theory} \label{sec:learning}
Our results have a number of consequences in computational learning theory. We discuss them below.

\paragraph{Technical Background: The Generalized Winnow Algorithm.}

The Generalized Winnow algorithm is one of the most powerful known algorithms for online learning \cite{littlestone, ksomb, thalercolt}. Suppose we are given a concept class $\mathcal{C}$ of functions mapping $n$-bit inputs to $\{-1, 1\}$, as well
as a collection of polynomial-time computable ``feature'' functions $\mathcal{F}$. The Generalized Winnow algorithm learns a concept in $\mathcal{C}$
by maintaining as a hypothesis a low-weight linear threshold function of features in $\mathcal{F}$.

Suppose that each $f \in \mathcal{C}$ has a low-weight linear threshold representation
$$f(x) = \operatorname{sgn}\left(\sum_{h_i \in \mathcal{F}} w_i h_i(x)\right),$$
where each $w_i$ is an integer, and $\sum_i |w_i| \leq W$.
A remarkable property of the Generalized Winnow algorithm is that
its mistake bound depends only \emph{logarithmically} on
the size of the feature set $\mathcal{F}$, and polynomially on the weight bound
$W$ (here the mistake bound refers to the worst-case number of mistakes an online learning algorithm makes over any sequence of examples). Meanwhile, its running time per example is polynomial in the size of the feature set. Standard techniques
can be used to transform any online learning algorithm into a PAC learning algorithm whose sample complexity is 
proportional to the mistake bound. 

\paragraph{PAC Learning \acz\ via Generalized Winnow.}
Valiant famously posed the problem of PAC learning
DNF formulas in his original
paper \cite{valiant} introducing the PAC model. The fastest known algorithm for this problem is due to Klivans and Servedio.
It is based on linear programming,
and takes time $\exp\left(\tilde{O}(n^{1/3})\right)$ \cite{ks}. At the core of this algorithm is a fundamental
structural result for DNFs: Klivans and Servedio
showed that every DNF of size $s$ can be computed
by a polynomial threshold function of degree $O(n^{1/3} \log s)$. 
However, the weight of the PTF arising in 
this construction can grow doubly-exponentially with $n$.
Klivans and Servedio asked whether it is possible that every
polynomial-size DNF has a PTF of degree $\tilde{O}(n^{1/3})$,
and weight $\exp\left(\tilde{O}(n^{1/3})\right)$ -- an affirmative
answer to this question would imply that the 
Generalized Winnow Algorithm (run over the feature set
of all low-degree parities) can also PAC learn DNFs in time  
 $\exp\left(\tilde{O}(n^{1/3})\right)$. Such a result would be attractive,
 as the Generalized Winnow algorithm is substantially simpler 
 than the linear programming algorithm of Klivans and Servedio.
 
 While we do not resolve the question of Klivans and Servedio
 for DNFs, we do resolve it in the negative for depth-three circuits.
 In fact, we rule out the possibility of 
 the Generalized Winnow algorithm PAC learning depth-three Boolean
 circuits in time $\exp\left(\tilde{O}(n^{2/5})\right)$ \emph{regardless}
 of the underlying feature set. That is,
 our lower bound holds even on feature sets
 that are not low-degree parities. 
 
 Specifically, Corollary \ref{cor:ed-ampthresh} implies the following
 result. The proof is identical to \cite[Theorem 8.1]{majmaj}
 and is omitted for brevity.
 
 \eat{\begin{corollary} The Generalized Winnow Algorithm
 requires time $\exp\left(\tilde{\Omega}(n^{2/5})\right)$
 to PAC learn the concept class of 
 polynomial-size depth-three Boolean circuits.
 \end{corollary}
 }
 
 \begin{corollary}
 Let $\mathcal{C}$ denote the concept class of polynomial-size depth-three
 Boolean circuits.
Let $\mathcal{F} = \{h_1, \dots, h_{m} : \{-1,1\}^n \rightarrow
\{-1, 1\}\}$ be arbitrary Boolean functions such that every
$f \in \mathcal{C}$ can be expressed as $f(x) = \operatorname{sgn}\left(\sum_{i=1}^{m} w_i h_i(x)\right)$ for some integers $w_1, \dots , w_m$ with
$|w_1| + \dots+ |w_m| \leq W$. Then
$m \cdot W > \exp\left(\tilde{\Omega}(n^{2/5})\right)$.
\end{corollary}

\paragraph{PAC Learning \acz\ via Boosting.}
While an $\exp\left(\tilde{\Omega}(n^{1/3})\right)$-time
algorithm is known for PAC learning polynomial-size DNF formulas,
no $\exp\left(o(n)\right)$-time algorithm is known 
even for learning polynomial-size depth-three Boolean circuits.
A natural approach to this problem is as follows.
Suppose that every function $f$ in a concept class $\mathcal{C}$ can be 
computed by a PTF (of arbitrary degree) over $\{0, 1\}^n$
with weight at most $W$. The well-known \emph{discriminator lemma}
 of Hajnal et al. \cite{hajnal} implies that under \emph{any} distribution, there is some conjunction (possibly
 of width $\Omega(n)$) that has correlation at least $1/W$ with $f$.
 One can then apply an agnostic learning algorithm
 for conjunctions (such as the $\exp\left(\tilde{O}(n^{1/2})\right)$-time
 polynomial regression algorithm of Kalai et al. \cite{kalaiagnostic}), combined
 with standard boosting techniques, to PAC-learn
 $\mathcal{C}$ in time polynomial in $\max\left(\exp\left(\tilde{O}(n^{1/2})\right), W\right)$.
 
 Thus, if one could prove an $\exp(\tilde{O}(n^{1/2}))$ upper bound
(for PTFs over $\{0, 1\}^n$) on the threshold weight of \acz, one would obtain an 
$\exp(\tilde{O}(n^{1/2}))$-time algorithm for PAC learning \acz. 
While our $\exp(\tilde{\Omega}(n^{2/5}))$ threshold weight lower bound for \acz\
does not rule out this possibility, it does establish new limitations for this technique. 
In particular, our threshold weight lower bound
implies that even if faster algorithms for agnostically learning
conjunctions are discovered, this boosting-based
approach to learning \acz\ cannot run in time
better than $\exp\left(\tilde{\Omega}(n^{2/5})\right)$.

\paragraph{Attribute-Efficient Learning.}
Attribute-efficient learning is a clean framework that captures the challenging and important problem of learning in the presence of irrelevant information \cite{blum}. 
A class $\mathcal{C}$ of Boolean functions over $\bits^n$ is said to be attribute-efficiently learnable if there is a $\poly(n)$-time online algorithm that learns any $f \in \mathcal{C}$ with mistake bound polynomial in the representation size of $f$. For example, the concept class
of read-once DNFs that depend on $k \ll n$ of their input
variables is attribute-efficiently learnable
if there is an online learning algorithm for this class
that runs in time $\poly(n)$ per example and
achieves mistake bound $\poly(k, \log n)$. 

Attribute-efficient learning is a challenging
problem, and many simple concept classes are not known
to be attribute-efficiently learnable, including decision lists and 
read-once DNFs. 
The Generalized Winnow algorithm, run over
the feature-space of low-degree parities, marks the best progress toward attribute-efficient learning of 
these concept classes (see e.g. \cite{ksomb, thalercolt}). Prior to our work,
it was unknown whether this approach could learn read-once
DNFs depending on $k$ variables in time $\exp\left(\tilde{O}(n^{1/3})\right)$ per example
and with mistake bound $\poly(k, \log n)$, as such a guarantee would
hold if every read-once DNF on $n$ variables were computed by a polynomial
threshold function of degree $\tilde{O}(n^{1/3})$ and weight $\poly(n)$.
Corollary \ref{cor:colt} rules out this possibility in a very strong sense,
as it implies the existence of 
a read-once DNF that cannot be computed by any PTF
of $\poly(n)$ weight, unless the degree is $\tilde{\Omega}(n)$.
\eat{\begin{corollary}
The Generalized Winnow algorithm, run over
the feature-space of parities of degree at most $d$,
requires time $\Omega({n \choose d})$ 
to achieve mistake bound $\max\left(\exp\left(\Omega(n/d^2)\right), \exp\left(\Omega(\sqrt{n/d})\right)\right)$.
\end{corollary}
}
Similarly, Corollary \ref{cor:ed-amp} establishes new limitations on
the efficiency of the Generalized Winnow algorithm in
the context of attribute-efficient learning of depth-three Boolean circuits.

\medskip
\noindent \textbf{Acknowledgements.} We are grateful 
to Sasha Sherstov, Robert \v{S}palek, Li-Yang Tan, and the anonymous reviewers for valuable feedback
on earlier versions of this manuscript.

\appendix

\section{Final Details of the Proof of \thmref{thm:big}} \label{app:proof1}

\subsection{Proof of Claim \ref{claim:show2}}
Let $\mu$ be the distribution on $\left(\{-1, 1\}^{m}\right)^t$ given by $\mu(x_1, \dots, x_t) = \prod_{i=1}^{t} |\psi(x_i)|$. Since $\psi$ is orthogonal to the constant polynomial, it has expected value 0, and hence the string $(\dots, \sgn(\psi(x_i)), \dots)$ is distributed uniformly in $\{-1, 1\}^{t}$
when one samples $(x_1, \dots, x_t)$ according to $\mu$. 
Thus,
$$ \sum_{(x_1, \dots, x_t) \in \left(\{-1, 1\}^{m}\right)^t} |\zeta(x_1, \dots, x_t)| = \sum_{z \in \{-1, 1\}^t} |\Psi(z)| = |\Psi(\mathbf{1})| + |\Psi(\mathbf{-1})| = 1,$$ proving \eqref{eq:show2}. 
\qed

\subsection{Proof of Claim \ref{claim:show3}}
We prove that the polynomial $\zeta$ defined in \eqref{eq:zeta} satisfies \eqref{eq:show3}, reproduced here for convenience.

$$ \sum_{(x_1, \dots, x_{t}) \in \left(\{-1, 1\}^{m}\right)^t}  \zeta(x_1, \dots, x_{t})\chi_S(x_1, \dots, x_t) =0   \text{ for each } |S| \le d.\quad \quad (\ref{eq:show3})$$

To prove \eqref{eq:show3}, notice that since $\Psi$ is orthogonal on $\{-1, 1\}^t$ to constant functions, we have the Fourier representation
$$\Psi(z) =	\sum_{\substack{T \subseteq \{1, \dots, t\} \\ T \neq \emptyset}}	\hat{\Psi}(T) \chi_T(z)$$
for some reals $\hat{\Psi}(T)$. We can thus write
\[\zeta(x_1, \dots, x_t) = 2^t \sum_{T \neq \emptyset} \hat{\Psi}(T) \prod_{i \in T} \psi(x_i)\prod_{i \notin T} |\psi(x_i)|.\]
Given a subset $S \subseteq \{1, \dots, t\} \times \{1, \dots, m\}$ with $|S| \le d$, partition $S = (\{1\} \times S_1) \cup \dots \cup (\{t\} \times S_t)$ where each $S_i \subseteq \{1, \dots, m\}$. Then
\begin{align*}
\sum_{(x_1, \dots, x_{t}) \in \left(\{-1, 1\}^{m}\right)^t}& \zeta(x_1, \dots, x_{t})\chi_S(x_1, \dots, x_t) \\
&= 2^t \sum_{T \neq \emptyset} \hat{\Psi}(T) \prod_{i \in T} \underbrace{\left( \sum_{x_i \in \{-1, 1\}^m} \psi(x_i) \chi_{S_i}(x_i)\right)} \prod_{i \notin T} \left(\sum_{x_i \in \{-1, 1\}^m} |\psi(x_i)| \chi_{S_i}(x_i)\right).
\end{align*}
Since $|S| \le d$, we have that $|S_i| \le d$ for every index $i \in \{1, \dots, t\}$. Thus for each set $T$, at least one of the underbraced factors is zero, as $\chi_{S_i}$ is orthogonal to $\psi$ whenever $|S_i| \le d$.

\section{Proof of Symmetrization Lemma \ref{lem:ambainis-sym}} \label{app:ambainis}

We now give a proof of Lemma \ref{lem:ambainis-sym}, which roughly shows that the symmetrization map $p \mapsto p^{\mathrm{sym}}$ does not increase the degree of $p$ by too much. The notation in the lemma and proof is defined in Section \ref{sec:intro-odeg-acz}. We also use the shorthand $\sigma \cdot x \cdot \pi$ to denote the boolean vector $y$ for which $g_y = \sigma \circ g_x \circ \pi$.

\lemambainissym*

The proof proceeds in three steps. In the first step, we perform a change of variables showing that $p(x)$ can be written as another polynomial $q(t)$, where $\deg q \le \deg p$. An input $t \in \{0, 1\}^{N \cdot R}$ to the new polynomial $q$ offers a different representation of a function $g_t : [N] \to [R]$ as follows: the variable $t_{ij} = 1$ if $g_t(i) = j$, and $t_{ij} = 0$ otherwise.

In the second step, we apply a lemma of Ambainis \cite{ambainis}, which shows that $q$ can be partially symmetrized to yield a polynomial $Q$ over yet a different set of variables, again without increasing its degree. This symmetrization yields a polynomial $Q$ whose input now represents a function in a manner invariant under permutations of the function's domain. Specifically, the inputs to the polynomial $Q$ are now variables $z_{j}$, where $z_j$ counts the number of $i \in [N]$ for which some function $g_z(i) = j$. Notice that if $g_w = g_z \circ \pi$ for a permutation $\pi$, then $w = z$ and hence $Q(w) = Q(z)$.

The third and final step is to symmetrize the polynomial $Q$ once again (without increasing its degree) so that it is invariant under permutations of its input variables $z_j$. Again interpreting each $z_j$ as the number of $i$ for which some function $g_z(i) = j$, the resulting polynomial $Q^{\mathrm{sym}}$ is now invariant under permutations of both the domain \emph{and} codomain of the function $g_z$.

However, in terms of the original variables $x$, the new variables $z_j$ are each polynomials of degree $\log_2 R$. Therefore, converting the fully symmetrized polynomial $Q^{\mathrm{sym}}$ back into a polynomial $p^{\mathrm{sym}}$ over $x$ potentially incurs a $\log_2 R$ factor blow-up in degree.

\begin{proof}
Let $p : \{-1, 1\}^m \to \R$ be a polynomial of degree $d$. Recall that our proof proceeds in three stages. In the first, we define a new polynomial $q$ over a different set of variables $t \in \{0, 1\}^{N \cdot R}$, and show that $\deg q \le \deg p$. To this end, define a map $T : \{-1, 1\}^m \to \{0, 1\}^{N\cdot R}$ by $T_{ij}(x) = 1$ if $g_x(i) = j$, and $T_{ij}(x) = 0$ otherwise. We claim that there is a polynomial $q : \{0, 1\}^{N \cdot R} \to \R$ of degree $d$ such that $q(T(x)) = p(x)$ for all $x \in \{-1, 1\}^m$. To see this, write $x$ as a list of blocks $x = (x_1, \dots, x_N)$ where each block has length $\log_2 R$, and let $x_{ik}$ denote the $k$'th bit of block $x_i$. Then
\[x_{ik} = 1 - 2\sum_{j : j_k = -1} T_{ij}(x),\]
where $j_k$ is the $k$'th bit of the binary representation of $j \in [N]$. Hence we can set
\[q(\dots, t_{ij}, \dots) = p \left(\dots,  1 - 2\sum_{j : j_k = -1} t_{ij}, \dots \right)\]
and have $q(T(x)) = p(x)$, where $\deg q \le \deg p = d$.

Recall that we can think of the variables $t_{ij}$ themselves as representing functions $g : [N] \to [R]$ whenever $t$ is the image of a boolean vector under the map $T$. That is, if $t = T(x)$ for some $x \in \{-1, 1\}^m$, then we can define the function $g_t: [N] \to [R]$ by $g_t \equiv g_x$. Specifically, $g_t(i)$ is the unique $j$ where $t_{ij} = 1$.

In the second step, we apply a lemma of Ambainis \cite{ambainis} shows that we can symmetrize the polynomial $q$, again without increasing its degree.

\begin{lemma}[\cite{ambainis}, Lemma 3.4]
Let $q : \{0, 1\}^{N \cdot R} \to \R$ be a polynomial of degree $d$ . Then there is a polynomial $Q: \{0, 1, \dots, N\}^R \to \R$ of degree $d$ such that
\[Q(\dots, \sum_{i = 1}^n t_{ij}, \dots) = \E_{\pi}[q(t \cdot \pi)]\]
whenever $t = T(x)$ for some $x \in \{-1, 1\}^m$. Here, $t \cdot \pi$ is shorthand for the $s$ for which $g_s = g_t \circ \pi$.
\end{lemma}

Applying the lemma yields a polynomial $Q$ such that 
\[Q(\dots, \sum_{i = 1}^n T_{ij}(x), \dots) = \E_{\pi}[q(T(x) \cdot \pi)] = \E_{\pi}[p(x \cdot \pi)].\]
For each $j = 1, \dots, R$, let $Z_{j}(x) = \sum_{i = 1}^n T_{ij}(x)$. Notice that $Z_{j}(x)$ counts the number of inputs $i \in [N]$ for which $g_x(i) = j$. This implies that for any permutation $\sigma$ on the range $[R]$, we also have that $Z_{\sigma^{-1}(j)}(x)$ counts the number of $i \in N$ for which $(\sigma \circ g_x)(i) = j$. Hence for any fixed $\sigma$,
\[Q(Z_{\sigma^{-1}(1)}(x), \dots, Z_{\sigma^{-1}(R)}(x)) = \E_\pi[p(\sigma \cdot x \cdot \pi)].\]

This observation allows us to complete the third step of the proof, which is symmetrization of the polynomial $Q$. That is, if we let
\[Q^{\mathrm{sym}}(z) = \E_{\sigma}[Q(z_{\sigma(1)}, \dots, z_{\sigma(R)})],\]
then
\[Q^{\mathrm{sym}}(Z_1(x), \dots, Z_R(x)) = \E_{\sigma, \pi}[p(\sigma \cdot x \cdot \pi)] = p^{\mathrm{sym}}(x).\]

Notice that since $Q$ is a polynomial of degree $d$, the symmetrization $Q^{\mathrm{sym}}$ clearly also has degree $d$. To complete the argument, we need to show that each function $Z_j(x)$ is a polynomial of degree $\log_2 R$. To see this, recall that $Z_j$ is a linear combination of the functions $T_{ij}$, so it suffices to show that $T_{ij}$ has degree $\log_2 R$. The function $T_{ij}(x)$ evaluates to $1$ if the block $x_i$ equals to binary representation of $j$, and evaluates to $0$ otherwise. Thus we can write
\[T_{ij}(x) = R\prod_{k = 1}^{\log_2 R} (j_kx_{ik} + 1),\]
where $x_{ik}$ is the $k$'th bit of block $x_i$, and $j_k$ is the $k$'th bit of the binary representation of $j$. This expression shows that $T_{ij}$ has degree $\log_2 R$, so the polynomial $p^{\mathrm{sym}}$ has degree $d \log_2 R$.

\eat{
Now write
\[p^{\mathrm{sym}}(x) = \E_{y \sim x} [p(y)] = \E_{y \sim x} [q(T(y))].\]
We need to show that $\deg p^{\mathrm{sym}} \le d$. Since $\deg q \le d$, by linearity of expectations it suffices to show that if $I$ is a single monomial in the variables $t_{ij}$ of degree $d$, then $\E[I(T(y))]$ also has degree $d$. To this end, let $I = t_{i_1j_1}t_{i_2j_2}\dots t_{i_dj_d}$ be a monomial. We may assume that the $i_\ell$ for $\ell = 1, 2, \dots, d$ are distinct, since any product $t_{ij}^2 = t_{ij}$ has a redundant variable, and $T_{ij}(x)T_{ij'}(x) = 0$ for any $j \ne j'$, since $g_x(i)$ cannot be both $j$ and $j'$ at the same time. We can then write
\[\E_{y \sim x}[I(T(y))] = \Pr_{y \sim x}[T_{i_1j_1}(y) = 1] \prod_{\ell = 2}^d \Pr_{y \sim x}[T_{i_\ell j_\ell}(y) = 1 | T_{i_1j_1}(y) = \dots = T_{i_{\ell-1} j_{\ell-1}}(y) = 1]\]
}
\end{proof}

\section{One-Sided Approximate Degree of \ED: Alternative Proof of Corollary \ref{cor:ed}}\label{app:as}

Improving on results of Aaronson and Shi \cite{aaronsonshi}, Ambainis \cite{ambainis} showed that the \ED\ problem with small range has approximate degree $\tilde{\Omega}(m^{2/3})$. Recall that the \ED\ problem on input size $m = N \log_2 N$, where $N$ is a power of $2$, takes as input $N$ blocks of length $\log_2 N$ and evaluates to $-1$ if and only if the blocks are distinct. We will show that there is a dual witness $\Psi$ for the high approximate degree of \ED\ having one-sided error. Hence,
this dual witness actually demonstrates that \ED\ has high
\emph{one-sided} approximate degree.

The idea is that any dual witness for \ED\ can be ``symmetrized'' to produce a new dual witness $\Psi$ that is constant on inputs $x \in T$, where $T$ is the set of inputs for which \ED\ evaluates to true. We then use the fact that $\Psi$ is balanced to argue
that the \emph{total} correlation of $\Psi$ with \ED\ is a constant
multiple of the correlation restricted to inputs in $T$. Since $\Psi$
has positive correlation with \ED, it follows that $\Psi$
must have the correct sign on all inputs in $T$, as desired.

Formally, let $\psi$ be a dual witness for the fact that $f = $ \ED\ has $\eps$-approximate degree $d = \tilde{\Omega}(m^{2/3})$ for some constant $\eps$. By \thmref{thm:prelim},

\begin{equation} \label{eq:ed0} \sum_{x \in \{-1, 1\}^m} f(x) \psi(x) > \eps, \end{equation}
\begin{equation} \label{eq:ed1} \sum_{x \in \{-1, 1\}^m} |\psi(x)| = 1,\end{equation}  and
\begin{equation} \label{eq:ed2} \sum_{x \in \{-1, 1\}^m} \psi(x) \chi_S(x)=0   \text{ for each } |S| \leq d.\end{equation}

For any permutation $\sigma \in S^N$, and $x = (x_1, \dots, x_N) \in \{-1, 1\}^m$, define
\[\sigma(x) = (x_{\sigma(1)}, \dots, x_{\sigma(N)}).\]
That is, $\sigma$ acts on $\{-1, 1\}^m$ by permuting the $N$ blocks of length $\log N$. Observe that for every permutation $\sigma$ and every $x \in \{-1, 1\}^m$, 
\begin{equation} \label{eq:ed-inv}
f(\sigma(x)) = f(x).
\end{equation}
Now define the symmetrized dual witness
\[\Psi(x) = \E_{\sigma \in S^N}[\psi(\sigma(x))].\]
We will show that $\Psi$ is a dual witness for $f$ with one-sided error by checking the conditions of \thmref{thm:oprelim}. First,
\begin{align*}
\sum_{x \in \{-1, 1\}^m} \Psi(x) f(x) &= \E_{\sigma \in S^N} \left[\sum_x \psi(\sigma(x)) f(x)\right]\\
&= \E_{\sigma \in S^N} \left[\sum_x \psi(x) f(x)\right] & \text{by \eqref{eq:ed-inv}}\\
&> \epsilon & \text{by (\ref{eq:ed0}),}
\end{align*}
verifying (\ref{eq:oprelim0}). Condition (\ref{eq:oprelim1}) is immediate from (\ref{eq:ed1}). Condition (\ref{eq:oprelim2}) follows because
\[\sum_{x \in \{-1, 1\}^m} \Psi(x) \chi_S(x) = \E_{\sigma \in S^N} \left[\sum_x \psi(x) \chi_{\sigma(S)}(x)\right]\]
where $\sigma(S) = \{\sigma(i) : i \in S\}$ and from (\ref{eq:ed2}).

Finally, we check the one-sided error condition (\ref{eq:oprelim3}). We will first show that $\Psi$ is constant on $f^{-1}(-1)$. Let $x^* = (x^*_1, \dots, x^*_N)$ where $x^*_i$ is the binary encoding of $i$. Since there are only $N$ distinct strings of length $\log N$, $f(x) = -1$ if and only if $x = \sigma_x(x^*)$ for some $\sigma_x \in S^N$. Therefore, if $f(x) = -1$, then
\[\Psi(x) = \E_{\sigma \in S^N}[\psi(\sigma(x))] = \E_{\sigma \in S^N}[\psi((\sigma \circ \sigma_x)(x^*))] = \Psi(x^*),\]
so $\Psi$ is constant on $f^{-1}(-1)$.

By condition (\ref{eq:oprelim0}) it holds that
\[\sum_{x \in f^{-1}(1)} \Psi(x) - \sum_{x \in f^{-1}(-1)} \Psi(x) > \eps,\]
 and by condition  (\ref{eq:oprelim1}) applied to $\chi_S$ for $S = \emptyset$ it holds that
\[\sum_{x \in f^{-1}(1)} \Psi(x) + \sum_{x \in f^{-1}(-1)} \Psi(x) = 0.\]
Subtracting the second equation from the first, we conclude that
\[-2 \sum_{x \in f^{-1}(-1)} \Psi(x) > \eps.\]
Since $\Psi$ is constant on $f^{-1}(-1)$, this implies that $\Psi(x) < 0$ whenever $x \in f^{-1}(-1)$, proving (\ref{eq:oprelim3}).

\section{Degree Independent Threshold Weight Bounds via Duality} \label{app:threshold-weight}

In this section, we use the dual characterization of threshold weight to give a new proof of a version of Krause's result translating degree-$d$ threshold weight lower bounds for a function $F$ into degree independent threshold weight lower bounds for a related function $F'$. Specifically, we prove the lemma

\lemthresholdweight*

\begin{proof}
By \thmref{thm:twdual} (condition (\ref{eq:twprelim0})), it suffices to exhibit a distribution $\mu'$ over $\bits^{3n}$ for which
\begin{equation*} |\E_{(x, y, z) \sim \mu'} [F'(x, y, z) \chi_S(x, y, z)]| \le \max \left\{\left(\frac{2n}{W(F, d)}\right)^{1/2}, 2^{-d/2}\right\} \text{ for all } S \subseteq \{1, \dots, 3n\}.\end{equation*}

We construct the distribution $\mu'$ as follows. By condition (\ref{eq:twprelim1}) of \thmref{thm:twdual}, there is a probability distribution $\mu$ over $\bits^n$ such that
\begin{equation} \label{eq:muinappendix} |\E_{w \sim \mu} [F(w) \chi_S(w)]| \le \left(\frac{2n}{W(F, d)}\right)^{1/2} \text{ for each } |S| \le d.\end{equation}
Define $\mu'(x, y, z) = 2^{-2n}\mu(\sel_z(x, y))$, where $\sel_z(x, y) = (\dots, (\bar{z}_i \land x_i) \lor (z_i \land y_i), \dots)$ selects for each index in $[n]$ a bit from either $x$ or $y$ according to $z$. The distribution $\mu'$ has a natural interpretation as follows: it first selects the string $z$ uniformly at random from $\{-1, 1\}^n$. Next, it sets the values of the variables in $(x, y)$ that are selected by $z$ so that they are distributed according to the distribution $\mu$. Finally, it sets the values of the unselected variables in $(x, y)$ uniformly at random. 

Note that $\mu'$ is indeed a probability distribution, as for every string $w \in \bits^n$, there are exactly $2^{2n}$ strings $(x, y, z)$ for which $\sel_z(x, y) = w$. Moreover, this observation allows us to write
\[\E_{(x, y, z) \sim \mu'} [F'(x, y, z) \chi_S(x, y, z)] = 2^{-2n} \sum_{w \in \bits^n} F(w) \mu(w) \sum_{(x, y, z): \sel_z(x, y) = w} \chi_S(x, y, z). \]
Write $S$ as the disjoint union $(\{1\} \times S_1) \cup (\{2\} \times S_2) \cup (\{3\} \times S_3)$ where $S_1, S_2, S_3$ correspond to indices in $x, y, z$ respectively. Then the expectation becomes
\[2^{-2n} \sum_{z \in \bits^n} \chi_{S_3}(z) \underbrace{\sum_{w \in \bits^n} F(w) \mu(w) \sum_{(x, y): \sel_z(x, y) = w} \chi_{S_1}(x)\chi_{S_2}(y)}\]
Let $G(z)$ denote the underbraced sum. 

Suppose there is an index $i \in S_3$ that is not contained in $S_1 \cup S_2$. Then for every $z \in \bits^n$, the string $z^i$ obtained from $z$ by flipping the bit at index $i$ satisfies $\chi_{S_3}(z^i) = -\chi_{S_3}(z)$. On the other hand,  
for any $(x, y) \in \{-1, 1\}^{2n}$, if we set $x' = (x_1, \dots, x_{i-1}, y_i, x_{i+1}, \dots, x_n)$ and analogously set $y' = (y_1, \dots, y_{i-1}, x_i, y_{i+1}, \dots, y_n)$, then $\sel_z(x, y) = \sel_{z^i}(x', y')$. Moreover,
because $i \not\in S_1 \cup S_2$, it holds that $\chi_{S_1}(x')\chi_{S_2}(y') = \chi_{S_1}(x)\chi_{S_2}(y)$. It follows that $G(z) = G(z^i)$, as each term $(x, y)$ in the underbraced sum defining $G(z)$ is ``matched''
by term $(x', y')$ in the underbraced sum defining $G(z^i)$. 
When combined with the fact that $\chi_{S_3}(z^i) = -\chi_{S_3}(z)$, we
 see that the terms corresponding to $z$ and $z^i$ in the outer sum 
 cancel out, and hence the entire outer sum evaluates to zero. We conclude that for the expectation to be nonzero, we must have $S_3 \subseteq S_1 \cup S_2$, and we assume this
 holds for the remainder of the proof.

Consider any $i \in S_1$. Then we claim that $G(z) = 0$ whenever $z_i$ selects $y_i$, i.e., for any $z$ such that $z_i = -1$. This can be seen by another pairing argument: If $\sel_z(x, y) = w$ but $z_i$ selects $y_i$, then $\sel_z(x^i, y) = w$ as well. However, $\chi_{S_1}(x) = - \chi_{S_1}(x^i)$ because $i \in S_1$. This ensures that the innermost sum is zero and hence $G(z) = 0$. The analogous statement holds also for any $i \in S_2$, so  for $G(z)$ to be nonzero, it must hold that
$z_i = 1$ for all $i \in S_1$ and $z_i = -1$ for all $i \in S_2$. Below, we refer to such a $z$ as a ``contributing'' $z$, and all other
values of $z$ as ``non-contributing''. In particular, we must have $S_1 \cap S_2 = \emptyset$ for $z$ to be contributing.

For any fixed contributing $z$, it holds that
\[\sum_{(x, y): \sel_z(x, y) = w} \chi_{S_1}(x)\chi_{S_2}(y) = 2^n\chi_{S_1 \cup S_2}(w).\]
Therefore, it holds that
\begin{align}
\notag
|\E_{(x, y, z) \sim \mu'} [F'(x, y, z) \chi_S(x, y, z)]| &= 2^{-2n}\left|\sum_{z \in \bits^n} \chi_{S_3}(z) G(z)\right| \\
&\notag \le 2^{-n} \sum_{z:G(z) \ne 0} \left|\sum_{w \in \bits^n} F(w) \mu(w) \chi_{S_1 \cup S_2}(w)\right| \\
&\le 2^{-|S_1| -|S_2|}\left|\sum_{w \in \bits^n} F(w) \mu(w) \chi_{S_1 \cup S_2}(w)\right|, \label{finaldamnineq}
\end{align}
where inequality (\ref{finaldamnineq}) used the fact that $G(z)=0$
for any non-contributing $z$.

Now we consider two cases for the size of $S$. First suppose $|S| \le d$, so in particular, $|S_1 \cup S_2| \le d$. Then \eqref{eq:muinappendix} and inequality (\ref{finaldamnineq}) implies that
\[|\E_{(x, y, z) \sim \mu'} [F'(x, y, z) \chi_S(x, y, z)]| \le \left(\frac{2n}{W(f, d)}\right)^{1/2}.\]
Second, suppose that $|S| > d$. We have argued that if
$\E_{(x, y, z) \sim \mu'} [F'(x, y, z) \chi_S(x, y, z)]  \neq 0$, then
 $S_3 \subseteq S_1 \cup S_2$.
Hence, it must be the case that $|S_1| + |S_2| \ge |S|/2 > d/2$. Therefore, 
inequality (\ref{finaldamnineq}) implies that $\E_{(x, y, z) \sim \mu'} [F'(x, y, z) \chi_S(x, y, z)]  \leq 2^{-d/2}$.
This completes the proof.

\end{proof}

\end{document}